\newif\ifdraft  \draftfalse   
\newif\ifpdflatex  \pdflatexfalse 
\documentclass[12pt]{article}
\usepackage[a4paper,margin=1.5in,top=1in,marginpar=1in]{geometry}
\usepackage{xspace}
\usepackage{vaucanson-g-O}
\usepackage{amssymb}
\usepackage{amsmath}
\usepackage{amsthm} 
\usepackage{amsfonts}
\usepackage{enumitem}
\usepackage{marginnote}
\usepackage{datetime}
\usepackage[mathlines]{lineno}
\usepackage[T1]{fontenc}

\newdateformat{isodate}{\THEYEAR--\twodigit{\THEMONTH}--\twodigit{\THEDAY}}
\isodate

\newcounter{tmpthm}
\title{Auto-similarity \\ 
       in rational base number systems}
\author
{
        Shigeki Akiyama\thanks{University of Tsukuba, 1-1-1 Tennodai, Tsukuba, Ibaraki, 350-8571 Japan}~
        \and
        Victor Marsault\thanks{Corresponding author, victor.marsault@telecom-paristech.fr}
        \thanks{Telecom-ParisTech and CNRS, 46 rue Barrault, 75013 Paris, France}~
        \and 
        Jacques Sakarovitch\footnotemark[\value{footnote}]
}

\date{\today}

\usepackage{ifthen}
\usepackage{etoolbox}
\usepackage{xargs}
\usepackage{slantsc}

\newcommandx{\newtheoremy}[3][2={}]{
  \ifthenelse{\equal{#2}{}}{
    \ifcsmacro{#1}{}{\newtheorem{#1}{#3}}
  }{
    \ifcsmacro{#1}{}{\newtheorem{#1}[#2]{#3}}
  }
}

\newtheoremy{thm}{thm}
\newcommand{\thmBlockFont}[1]{\sc{#1}}
\newcommand{\RefSeparator}{.}
\newcommand{\generalref}[2]{\ref{#1\RefSeparator#2}}
\newcommand{\generalpageref}[2]{\pageref{#1\RefSeparator#2}}
\newcommand{\generallabel}[2]{\label{#1\RefSeparator#2}}
\newcommand{\PageName}{page}

\newcommand{\AlgorithmName}{Algorithm}
\newcommand{\AlgorithmRefName}{Algorithm}
\newcommand{\AlgorithmRefPrefix}{a}
\newtheoremy{algorithm}[thm]{\thmBlockFont{\AlgorithmName}}

\makeatletter
\newcommand*{\ralgorithm}{\@ifstar{\generalref{\AlgorithmRefPrefix}}{\AlgorithmRefName~\ralgorithm*}}
\newcommand*{\palgorithm}{\@ifstar{\generalpageref{\AlgorithmRefPrefix}}{\PageName~\palgorithm*}}
\makeatother

\newcommand{\CorollaryName}{Corollary}
\newcommand{\CorollaryRefName}{Corollary}
\newcommand{\CorollaryRefPrefix}{c}
\newtheoremy{corollary}[thm]{\thmBlockFont{\CorollaryName}}
\newcommand{\lcorollary}[1]{\generallabel{\CorollaryRefPrefix}{#1}}
\makeatletter
\newcommand*{\rcorollary}{\@ifstar{\generalref{\CorollaryRefPrefix}}{\CorollaryRefName~\rcorollary*}}
\newcommand*{\pcorollary}{\@ifstar{\generalpageref{\CorollaryRefPrefix}}{\PageName~\pcorollary*}}
\makeatother

\newcommand{\ConjectureName}{Conjecture}
\newcommand{\ConjectureRefName}{Conjecture}
\newcommand{\ConjectureRefPrefix}{cj}
\newtheoremy{conjecture}[thm]{\thmBlockFont{\ConjectureName}}

\makeatletter
\newcommand*{\rconjecture}{\@ifstar{\generalref{\ConjectureRefPrefix}}{\ConjectureRefName~\rconjecture*}}
\newcommand*{\pconjecture}{\@ifstar{\generalpageref{\ConjectureRefPrefix}}{\PageName~\pconjecture*}}
\makeatother

\newcommand{\DefinitionName}{Definition}
\newcommand{\DefinitionRefName}{Definition}
\newcommand{\DefinitionRefPrefix}{d}
\newtheoremy{definition}[thm]{\thmBlockFont{\DefinitionName}}
\newcommand{\ldefinition}[1]{\generallabel{\DefinitionRefPrefix}{#1}}
\makeatletter
\newcommand*{\rdefinition}{\@ifstar{\generalref{\DefinitionRefPrefix}}{\DefinitionRefName~\rdefinition*}}
\newcommand*{\pdefinition}{\@ifstar{\generalpageref{\DefinitionRefPrefix}}{\PageName~\pdefinition*}}
\makeatother

\newcommand{\ExampleName}{Example}
\newcommand{\ExampleRefName}{Example}
\newcommand{\ExampleRefPrefix}{e}
\newtheoremy{example}[thm]{\thmBlockFont{\ExampleName}}

\makeatletter
\newcommand*{\rexample}{\@ifstar{\generalref{\ExampleRefPrefix}}{\ExampleRefName~\rexample*}}
\newcommand*{\pexample}{\@ifstar{\generalpageref{\ExampleRefPrefix}}{\PageName~\pexample*}}
\makeatother

\newcommand{\LemmaName}{Lemma}
\newcommand{\LemmaRefName}{Lemma}
\newcommand{\LemmaRefPrefix}{l}
\newtheoremy{lemma}[thm]{\thmBlockFont{\LemmaName}}
\newcommand{\llemma}[1]{\generallabel{\LemmaRefPrefix}{#1}}
\makeatletter
\newcommand*{\rlemma}{\@ifstar{\generalref{\LemmaRefPrefix}}{\LemmaRefName~\rlemma*}}
\newcommand*{\plemmam}{\@ifstar{\generalpageref{\LemmaRefPrefix}}{\PageName~\plemma*}}
\makeatother

\newcommand{\PropositionName}{Proposition}
\newcommand{\PropositionRefName}{Proposition}
\newcommand{\PropositionRefPrefix}{p}
\newtheoremy{proposition}[thm]{\thmBlockFont{\PropositionName}}
\newcommand{\lproposition}[1]{\generallabel{\PropositionRefPrefix}{#1}}
\makeatletter
\newcommand*{\rproposition}{\@ifstar{\generalref{\PropositionRefPrefix}}{\PropositionRefName~\rproposition*}}
\newcommand*{\pproposition}{\@ifstar{\generalpageref{\PropositionRefPrefix}}{\PageName~\pproposition*}}
\makeatother

\newcommand{\PropertyName}{Property}
\newcommand{\PropertyRefName}{Property}
\newcommand{\PropertyRefPrefix}{pp}
\newtheoremy{property}[thm]{\thmBlockFont{\PropertyName}}

\makeatletter
\newcommand*{\rproperty}{\@ifstar{\generalref{\PropertyRefPrefix}}{\PropertyRefName~\rproperty*}}
\newcommand*{\pproperty}{\@ifstar{\generalpageref{\PropertyRefPrefix}}{\PageName~\pproperty*}}
\makeatother

\newcommand{\QuestionName}{Question}
\newcommand{\QuestionRefName}{Question}
\newcommand{\QuestionRefPrefix}{q}
\newtheoremy{question}[thm]{\thmBlockFont{\QuestionName}}

\makeatletter
\newcommand*{\rquestion}{\@ifstar{\generalref{\QuestionRefPrefix}}{\QuestionRefName~\rquestion*}}
\newcommand*{\pquestion}{\@ifstar{\generalpageref{\QuestionRefPrefix}}{\PageName~\pquestion*}}
\makeatother

\newcommand{\RemarkName}{Remark}
\newcommand{\RemarkRefName}{Remark}
\newcommand{\RemarkRefPrefix}{r}
\newtheoremy{remark}[thm]{\thmBlockFont{\RemarkName}}
\newcommand{\lremark}[1]{\generallabel{\RemarkRefPrefix}{#1}}
\makeatletter
\newcommand*{\rremark}{\@ifstar{\generalref{\RemarkRefPrefix}}{\RemarkRefName~\rremark*}}
\newcommand*{\premark}{\@ifstar{\generalpageref{\RemarkRefPrefix}}{\PageName~\premark*}}
\makeatother

\newcommand{\NotationName}{Notation}
\newcommand{\NotationRefName}{Notation}
\newcommand{\NotationRefPrefix}{n}
\newtheoremy{notation}[thm]{\thmBlockFont{\NotationName}}

\makeatletter
\newcommand*{\rnotation}{\@ifstar{\generalref{\NotationRefPrefix}}{\NotationRefName~\rnotation*}}
\newcommand*{\pnotation}{\@ifstar{\generalpageref{\NotationRefPrefix}}{\PageName~\pnotation*}}
\makeatother

\newcommand{\TheoremName}{Theorem}
\newcommand{\TheoremRefName}{Theorem}
\newcommand{\TheoremRefPrefix}{t}
\newtheoremy{theorem}[thm]{\thmBlockFont{\TheoremName}}
\newtheoremy{Theorem}{\thmBlockFont{\TheoremName}}

\newcommand{\ltheorem}[1]{\generallabel{\TheoremRefPrefix}{#1}}
\makeatletter
\newcommand*{\rtheorem}{\@ifstar{\generalref{\TheoremRefPrefix}}{\TheoremRefName~\rtheorem*}}
\newcommand*{\ptheorem}{\@ifstar{\generalpageref{\TheoremRefPrefix}}{\PageName~\ptheorem*}}
\makeatother

\newcommand{\FigureRefName}{Figure}
\newcommand{\FigureRefPrefix}{f}
\newcommand{\lfigure}[1]{\generallabel{\FigureRefPrefix}{#1}}
\makeatletter
\newcommand*{\rfigure}{\@ifstar{\generalref{\FigureRefPrefix}}{\FigureRefName~\rfigure*}}
\newcommand*{\pfigure}{\@ifstar{\generalpageref{\FigureRefPrefix}}{\PageName~\pfigure*}}
\makeatother

\newcommand{\EquationRefName}{Equation}
\newcommand{\EquationRefPrefix}{eq}

\makeatletter
\newcommand*{\requation}{\@ifstar{\generalref{\EquationRefPrefix}}{\EquationRefName~\requation*}}
\newcommand*{\pequation}{\@ifstar{\generalpageref{\EquationRefPrefix}}{\PageName~\pequation*}}
\makeatother

\newcommand{\SectionRefName}{Section}
\newcommand{\SectionRefPrefix}{s}
\newcommand{\lsection}[1]{\generallabel{\SectionRefPrefix}{#1}}
\makeatletter
\newcommand*{\rsection}{\@ifstar{\generalref{\SectionRefPrefix}}{\SectionRefName~\rsection*}}
\newcommand*{\psection}{\@ifstar{\generalpageref{\SectionRefPrefix}}{\PageName~\psection*}}
\makeatother

\newcommand{\ProblemName}{Problem}
\newcommand{\ProblemRefName}{Problem}
\newcommand{\ProblemRefPrefix}{pb}
\newtheoremy{problem}[thm]{\thmBlockFont{\ProblemName}}

\makeatletter
\newcommand*{\rproblem}{\@ifstar{\generalref{\ProblemRefPrefix}}{\ProblemRefName~\rproblem*}}
\newcommand*{\pproblem}{\@ifstar{\generalpageref{\ProblemRefPrefix}}{\PageName~\pproblem*}}
\makeatother

\newcommand{\vmrefprefix}[1]{%
  \ifthenelse{\equal{#1}{corollary}}{\CorollaryRefPrefix}{}%
  \ifthenelse{\equal{#1}{definition}}{\DefinitionRefPrefix}{}%
  \ifthenelse{\equal{#1}{example}}{\ExampleRefPrefix}{}%
  \ifthenelse{\equal{#1}{lemma}}{\LemmaRefPrefix}{}%
  \ifthenelse{\equal{#1}{proposition}}{\PropositionRefPrefix}{}%
  \ifthenelse{\equal{#1}{property}}{\PropertyRefPrefix}{}%
  \ifthenelse{\equal{#1}{question}}{\QuestionRefPrefix}{}%
  \ifthenelse{\equal{#1}{remark}}{\RemarkRefPrefix}{}%
  \ifthenelse{\equal{#1}{notation}}{\NotationRefPrefix}{}%
  \ifthenelse{\equal{#1}{theorem}}{\TheoremRefPrefix}{}%
  \ifthenelse{\equal{#1}{figure}}{\FigureRefPrefix}{}%
  \ifthenelse{\equal{#1}{equation}}{\EquationRefPrefix}{}%
  \ifthenelse{\equal{#1}{section}}{\SectionRefPrefix}{}%
}

\newcommand{\vmrefname}[1]{
  \ifthenelse{\equal{#1}{corollary}}{\CorollaryRefName}{}%
  \ifthenelse{\equal{#1}{definition}}{\DefinitionRefName}{}%
  \ifthenelse{\equal{#1}{example}}{\ExampleRefName}{}%
  \ifthenelse{\equal{#1}{lemma}}{\LemmaRefName}{}%
  \ifthenelse{\equal{#1}{proposition}}{\PropositionRefName}{}%
  \ifthenelse{\equal{#1}{property}}{\PropertyRefName}{}%
  \ifthenelse{\equal{#1}{question}}{\QuestionRefName}{}%
  \ifthenelse{\equal{#1}{remark}}{\RemarkRefName}{}%
  \ifthenelse{\equal{#1}{notation}}{\NotationRefName}{}%
  \ifthenelse{\equal{#1}{theorem}}{\TheoremRefName}{}%
  \ifthenelse{\equal{#1}{figure}}{\FigureRefName}{}%
  \ifthenelse{\equal{#1}{equation}}{\EquationRefName}{}%
  \ifthenelse{\equal{#1}{section}}{\SectionRefName}{}%
}

  \usepackage{atbegshi,picture}
\def\Vhrulefill{\leavevmode\leaders\hrule height 0.7ex depth \dimexpr0.4pt-0.7ex\hfill\kern0pt}

\def\jscompatibility{0}
\def\curlang{en}

\newcommandx{\vmnewcommandx}[5][2=0,3={},5={},usedefault]{
      \ifthenelse{\equal{\jscompatibility}{0}}
      {\newcommandx{#1}[#2][#3]{#4}} 
      {\newcommandx{#1}[#2][#3]{#5}} 
}

\vmnewcommandx{\wlen}[1]{|#1|}

\vmnewcommandx{\cod}[1]{\langle #1 \rangle}
\vmnewcommandx{\floor}[1]{\lfloor #1 \rfloor}
\vmnewcommandx{\ceil}[1]{\lceil #1 \rceil}

\newcommandx{\newcommandy}[5][1=i,3=0,4={}]{%
  \ifthenelse{\isundefined{#2}}{\newcommandx{#2}[#3][#4]{#5}}{%
      \ifthenelse{\equal{#1}{i}}{}{}%
      \ifthenelse{\equal{#1}{o}}{\renewcommandx{#2}[#3][#4]{#5}}{}%
    }%
}
\newcommandy[o]{\pathx}[2][2={}]{
  \nlb%
  \ifthenelse{\equal{#2}{}}%
    {\xrightarrow{\ #1 \ }}%
    {\underset{#2}{\xrightarrow{\ #1 \ }}}%
  \nlb%
}
\newcommandy[o]{\pathy}[2][2={}]{\nlb\underset{#2}{\xleftarrow{\ #1 \ }}\nlb}
\newcommand{\transpair}[2]{ #1 \xmd | \xmd #2 }

\newcommand{\val}[1]{\widebar{#1}}
\newcommand{\card}[1]{|#1|}

\newcommand{\ssc}[1]{\textbf{\textsc{#1}}}

\newcommand{\set}[1]{\{#1\}}

\newcommand{\Z}{\mathbb{Z}}
\newcommand{\N}{\mathbb{N}}

\newcommand{\R}{\mathbb{R}}
\newcommand{\widebar}{\overline}

\newcommandy[o]{\Cup}{\bigcup}
\newcommand{\nlb}{\nolinebreak}

\renewcommand{\thmBlockFont}[1]{\ssc{#1}}

\newcommand{\frcurrentcaption}{}
\newcommand{\encurrentcaption}{}
\newcommandy{\vmfigure}[2]{
  \begin{figure}[ht!]
    \centering
    \input{#1/#2}
    \ifthenelse{\equal{\curlang}{fr}}
    {\caption{\frcurrentcaption}}{}
    \ifthenelse{\equal{\curlang}{en}}
    {\caption{\encurrentcaption}}{}%
    \def\vmlastfigure{#2}
    \lfigure{#2}
  \end{figure}
}

\newcommand{\vmStartTrickI}[3]{#1{#3}#2}
\newcommand{\vmStarTrickII}[4]{#1{#3}{#4}#2}

\makeatletter
\newcommand*{\addStartTextModeZ}[3][i]{ 
  \newcommandy[#1]{#2}{\protect\@ifstar{\leavevmode\protect\nlb$\protect#2$}{\protect#3}}
}
\newcommand*{\addStartTextModeI}[2]{
  \newcommand{#1}{\@ifstar{\leavevmode\nlb$\vmStartTrickI{#1}{$}}{#2}}
}
\newcommand*{\addStartTextModeII}[2]{
  \newcommand{#1}{\@ifstar{\leavevmode\nlb$\vmStarTrickII{#1}{$}}{#2}}
}
\newcommand*{\addStarTextModeIII}[2]{
  \newcommand{#1}{\@ifstar{\leavevmode\nlb$\vmStarTrickIII{#1}{$}}{#2}}
} 

\newcommand*{\addMagicMathModeZ}[3][i]{
  \newcommandy[#1]{#2}{{\ifmmode #3 \else \leavevmode\protect\nlb$\protect#3$\fi}}
}

\makeatletter

\renewcommand{\leq}{\leqslant}
\renewcommand{\geq}{\geqslant}
\renewcommand{\phi}{\varphi}
\renewcommand{\epsilon}{\varepsilon}
\renewcommand{\mod}{\text{~mod~}}

\newcommand{\e}{\text{\quad}}                 
\newcommand{\ee}{\text{\qquad}}               
\newsavebox{\InterSymbolSpace}
\savebox{\InterSymbolSpace}{\hspace{0.125em}}
\newsavebox{\SideFormulaSpace}
\savebox{\SideFormulaSpace}{\hspace{0.2em}}
\newcommand{\msp}{\usebox{\SideFormulaSpace}} 
\newcommand{\xmd}{\usebox{\InterSymbolSpace}} 
\newcommand{\eqpnt}{\makebox[0pt][l]{\: .}}

\newcommand{\eqvrg}{\makebox[0pt][l]{\: ,}}

\newcommand{\quantvrg}{\, , \;}
\newcommand{\quantsp}{\ee }

\newcommand{\LatinLocution}[1]{{\itshape #1}\xspace}

\newcommand{\cf}{\LatinLocution{cf.}}

\newcommand{\ie}{{that is, }}

\newcommand{\via}{via\xspace}

\newcommand{\UNmbb}{{\mathchoice
{\hbox{$\textstyle\rm 1\kern-0.2em I$}}%
{\hbox{$\textstyle\rm 1\kern-0.2em I$}}%
{\hbox{$\scriptstyle\rm 1\kern-0.15em I$}}%
{\hbox{$\scriptscriptstyle\rm 1\kern-0.1em I$}}%
}}
\newcommand{\Ac}{\mathcal{A}}

\newcommand{\Dc}{\mathcal{D}}

\newcommand{\Lc}{\mathcal{L}}

\newcommand{\Tc}{\mathcal{T}}

\newcommand{\Xc}{\mathcal{X}}
\newcommand{\Yc}{\mathcal{Y}}

\newlength{\ArrowDiagSize}
\newlength{\ArrowDiagWidth}
\newpsstyle{SLDiagStyle}%
   {colsep=6ex,rowsep=5ex,nodesep=1ex,npos=.45,%
    arrows=->,linewidth=\ArrowDiagWidth,arrowsize=\ArrowDiagSize,%
        linestyle=solid,linecolor=\ArrowDiagColor}

\newenvironment{SLDiag}%
   {\psset{style=SLDiagStyle}\begin{psmatrix}}%
   {\end{psmatrix}}%
\newcommand{\CDSL}{\begin{SLDiag}}
\newcommand{\CDSLF}{\end{SLDiag}}
\newenvironment{DiagraBig}%
{\psmatrix[colsep=7ex,rowsep=6ex,arrows=->,nodesep=1ex,npos=.45]}%
{\endpsmatrix}
\newcommand{\CDB}{\begin{DiagraBig}}
\newcommand{\CDBF}{\end{DiagraBig}}
\newenvironment{DiagraSmall}%
{\psmatrix[colsep=3ex,rowsep=3ex,arrows=->,nodesep=1ex,npos=.45]}%
{\endpsmatrix}
\newcommand{\CDS}{\begin{DiagraSmall}}
\newcommand{\CDSF}{\end{DiagraSmall}}

\newcommand{\matriceuu}[1]%
    {\begin{pmatrix} #1 \end{pmatrix}}
\newcommand{\matricedd}[4]%
    {\begin{pmatrix} #1 & #2 \\ #3 & #4 \end{pmatrix}}
\newcommand{\vecteurd}[2]%
    {\begin{pmatrix} #1 \\ #2 \end{pmatrix}}
\newcommand{\ligned}[2]%
    {\begin{pmatrix} #1 & #2 \end{pmatrix}}
\newcommand{\matricett}[9]%
    {\begin{pmatrix}  #1 & #2 & #3 \\
                      #4 & #5 & #6 \\
                      #7 & #8 & #9 \end{pmatrix}}
\newcommand{\vecteurt}[3]%
    {\begin{pmatrix} #1 \\ #2 \\ #3 \end{pmatrix}}
\newcommand{\lignet}[3]%
    {\begin{pmatrix} #1 & #2 & #3 \end{pmatrix}}

\newlength{\jsWidthCol}
\newlength{\blocinterligne}
\newlength{\blocinterligned}
\newlength{\temparraycolsep}
\newlength{\longueurbloc}
\newlength{\hauteurbloc}
\newlength{\centragebloc}
\newlength{\longueurblc}
\newlength{\hauteurblc}
\newlength{\centrageblc}
\newcommand{\blocligne}[1]%
    {\framebox[\longueurbloc]{$#1$}}
\newcommand{\blocmatrice}[1]%
    {\framebox[\longueurbloc]{\rule[\centragebloc]{0mm}{\hauteurbloc}$#1$}}
\newcommand{\blocvecteur}[1]%
    {\framebox{\rule[\centragebloc]{0mm}{\hauteurbloc}$#1$}}
\newcommand{\blcligne}[1]%
    {\framebox[\longueurblc]{$#1$}}
\newcommand{\blcmatrice}[1]%
    {\framebox[\longueurblc]{\rule[\centrageblc]{0mm}{\hauteurblc}$#1$}}
\newcommand{\blcvecteur}[1]%
    {\framebox{\rule[\centrageblc]{0mm}{\hauteurblc}$#1$}}
\newcommand{\matriceddblvs}[4]
   {\setlength{\temparraycolsep}{\arraycolsep}%
    \setlength{\arraycolsep}{1.3pt}%
        \left (%
    \begin{array}{cc}%
                #1  & \blcligne{#2} \\
            \blcvecteur{#3} & \blcmatrice{#4}
        \end{array}%
        \right )%
    \setlength{\arraycolsep}{\temparraycolsep}%
   }%
\newcommand{\vecteurdblvs}[2]%
   {\setlength{\temparraycolsep}{\arraycolsep}%
    \setlength{\arraycolsep}{1.5pt}%
        \left (%
    \begin{array}{c}%
                #1  \\
            \blcvecteur{#2}
        \end{array}%
        \right )%
    \setlength{\arraycolsep}{\temparraycolsep}%
   }%
\newcommand{\lignedblvs}[2]%
   {\setlength{\temparraycolsep}{\arraycolsep}%
    \setlength{\arraycolsep}{1.5pt}%
        \left (%
    \begin{array}{cc}%
                #1  & \blcligne{#2}
        \end{array}%
        \right )%
    \setlength{\arraycolsep}{\temparraycolsep}%
   }%
\newcommand{\matricettblvs}[9]
   {\setlength{\temparraycolsep}{\arraycolsep}%
    \setlength{\arraycolsep}{1.5pt}%
        \left (%
    \begin{array}{ccc}%
                #1  & \blcligne{#2} & #3\\
            \blcvecteur{#4} & \blcmatrice{#5} & \blcvecteur{#6}\\
                #7  & \blcligne{#8} & #9\\
        \end{array}%
        \right )%
    \setlength{\arraycolsep}{\temparraycolsep}%
   }%
\newcommand{\vecteurtblvs}[3]%
   {\setlength{\temparraycolsep}{\arraycolsep}%
    \setlength{\arraycolsep}{1.5pt}%
        \left (%
    \begin{array}{c}%
                #1  \\
            \blcvecteur{#2}\\
                #3
        \end{array}%
        \right )%
    \setlength{\arraycolsep}{\temparraycolsep}%
   }%
\newcommand{\lignetblvs}[3]%
   {\setlength{\temparraycolsep}{\arraycolsep}%
    \setlength{\arraycolsep}{1.5pt}%
        \left (%
    \begin{array}{ccc}%
                #1  & \blcligne{#2} & #3
        \end{array}%
        \right )%
    \setlength{\arraycolsep}{\temparraycolsep}%
   }%
\newcommand{\matricettblblvs}[9]
   {\setlength{\temparraycolsep}{\arraycolsep}%
    \setlength{\arraycolsep}{1.5pt}%
        \left (%
    \begin{array}{ccc}%
                #1  & \blcligne{#2} & \blcligne{#3}\\
            \blcvecteur{#4} & \blcmatrice{#5} & \blcmatrice{#6}\\
                \blcvecteur{#7}  & \blcmatrice{#8} & \blcmatrice{#9}\\
        \end{array}%
        \right )%
    \setlength{\arraycolsep}{\temparraycolsep}%
   }%
\newcommand{\vecteurtblblvs}[3]%
   {\setlength{\temparraycolsep}{\arraycolsep}%
    \setlength{\arraycolsep}{1.5pt}%
        \left (%
    \begin{array}{c}%
                #1  \\
            \blcvecteur{#2}\\
                \blcvecteur{#3}
        \end{array}%
        \right )%
    \setlength{\arraycolsep}{\temparraycolsep}%
   }%
\newcommand{\lignetblblvs}[3]%
   {\setlength{\temparraycolsep}{\arraycolsep}%
    \setlength{\arraycolsep}{1.5pt}%
        \left (%
    \begin{array}{ccc}%
                #1  & \blcligne{#2} & \blcligne{#3}
        \end{array}%
        \right )%
    \setlength{\arraycolsep}{\temparraycolsep}%
   }%
\newlength{\DefiTest}
\newlength{\DefiHeightu}\newlength{\DefiHeightd}%
\newlength{\DefiDepthu}\newlength{\DefiDepthd}%
\newcommand{\Defi}[2]%
    {%
     \settoheight{\DefiHeightu}{${\displaystyle #1}$}%
     \settodepth{\DefiDepthu}{${\displaystyle #1}$}%
     \addtolength{\DefiHeightu}{\DefiDepthu}%
     \settoheight{\DefiHeightd}{${\displaystyle #2}$}%
     \settodepth{\DefiDepthd}{${\displaystyle #2}$}%
     \addtolength{\DefiHeightd}{\DefiDepthd}%
     \left\{#1%
     \rule[-\DefiDepthd]{\DefiTest}{\DefiHeightd}%
     \xmd\right|%
     \left.%
     \rule[-\DefiDepthu]{\DefiTest}{\DefiHeightu}%
      #2\right\}%
     }
\newlength{\ColoText}
\newlength{\ColoFigu}
\newlength{\TextFiguSpace}
\newlength{\parindenttemp} 
\newlength{\parskiptemp} 
\newlength{\fboxseptemp} 
\newcommand{\TFBoxing}{}
\newcommand{\TFVertAlig}{}
\newcommand{\LeftLarg}{}
\renewcommand{\LeftLarg}{.66}
\ifdraft\renewcommand{\TFBoxing}{\fbox}\fi
\newcommand{\TxtFg}[3]%
   {%
    \setlength{\ColoText}{#1\textwidth}%
    \setlength{\ColoFigu}{\textwidth}%
    \addtolength{\ColoFigu}{-\ColoText}%
    \addtolength{\ColoText}{-.5\TextFiguSpace}%
    \addtolength{\ColoFigu}{-.5\TextFiguSpace}%
    \ifdraft\setlength{\fboxsep}{0pt}\fi
    \noi
    \TFBoxing{%
       \begin{minipage}[\TFVertAlig]{\ColoText}%
          \setlength{\parindent}{\parindenttemp}%
          \setlength{\parskip}{\parskiptemp}%
          \par\vspace*{0mm}
             #2
       \end{minipage}%
             }%
    \hspace*{\TextFiguSpace}%
    \TFBoxing{%
       \begin{minipage}[\TFVertAlig]{\ColoFigu}%
          \par\vspace*{0mm}%
             #3%
       \end{minipage}%
             }%
    \ifdraft\setlength{\fboxsep}{\fboxseptemp}\fi
   }%
\newcommand{\TextFigu}[3][\LeftLarg]%
   {\renewcommand{\TFVertAlig}{t}\TxtFg{#1}{#2}{#3}}
\newcommand{\TextFiguC}[3][\LeftLarg]%
   {\renewcommand{\TFVertAlig}{c}\TxtFg{#1}{#2}{#3}}
\newcommand{\TextFiguX}[3][\LeftLarg]
   {%
    \setlength{\ColoText}{#1\textwidth}%
    \setlength{\ColoFigu}{\textwidth}%
    \addtolength{\ColoFigu}{-\ColoText}%
    \addtolength{\ColoText}{-.5\TextFiguSpace}%
    \addtolength{\ColoFigu}{-.5\TextFiguSpace}%
    \addtolength{\ColoFigu}{\ETAExtendedLineWidth}
    \ifdraft\setlength{\fboxsep}{0pt}\fi
    \noi
    \ifodd\value{page}%
       \TFBoxing{%
          \begin{minipage}[t]{\ColoText}%
             \RstBLS
             \setlength{\parindent}{\parindenttemp}%
             \setlength{\parskip}{\parskiptemp}%
             \par\vspace*{0mm}
                #2
          \end{minipage}%
                }%
       \hspace*{\TextFiguSpace}%
       \TFBoxing{%
          \begin{minipage}[t]{\ColoFigu}%
             \par\vspace*{0mm}%
                #3%
          \end{minipage}%
                }%
    \else
       \hspace*{-\ETAExtendedLineWidth}
       \TFBoxing{%
          \begin{minipage}[t]{\ColoFigu}%
             \par\vspace*{0mm}%
                #3%
          \end{minipage}%
                }%
       \hspace*{\TextFiguSpace}%
       \TFBoxing{%
          \begin{minipage}[t]{\ColoText}%
             \RstBLS
             \setlength{\parindent}{\parindenttemp}%
             \setlength{\parskip}{\parskiptemp}%
             \par\vspace*{0mm}
                #2
          \end{minipage}%
                }%
    \fi%
    \ifdraft\setlength{\fboxsep}{\fboxseptemp}\fi
   }
\newcommand{\Axio}[1]%
   {\pointn #1\hspace*{.1em}\jspointtiret\hspace*{.4em}\ignorespaces}

\newlength{\sttd}
\newcommand{\ChgStateDiameter}[1]%
    {\setlength{\sttd}{\StateDiam}\FixStateDiameter{#1\sttd}}
\newcommand{\EdgeLabelPos}{}
\newcommand{\ChgEdgeLabelPosit}[1]{\renewcommand{\EdgeLabelPos}{#1}}
\newcommand{\RstEdgeLabelPosit}{\ChgEdgeLabelPosit{\EdgeLabelPosit}}
\newcommand{\SetEdgeLabelPosit}[1]%
   {\renewcommand{\EdgeLabelPosit}{#1}\RstEdgeLabelPosit}
   \renewcommand{\ChgTransLabelSep}[1]
      {\setlength{\TransLabelSP}{#1\TransLabelSep}}
   \renewcommand{\RstTransLabelSep}{\ChgTransLabelSep{1}}
   \renewcommand{\SetTransLabelSep}[1]
      {\setlength{\TransLabelSep}{#1}\RstTransLabelSep}
   \renewcommand{\ChgRelTransLabelSep}[1]
         {\setlength{\TransLabelSP}{#1\TransLabelSP}}
\newcommand{\ChgRelStateLineWidth}[1]%
   {\setlength{\StateLineWid}{#1\StateLineWid}}%
\renewcommand{\ChgRelEdgeLineWidth}[1]
   {\setlength{\EdgeLineWid}{#1\EdgeLineWid}}
\renewcommand{\VaucEdgeLabel}[1]%
       {\textcolor{\EdgeLabelCol}%

{\scalebox{\EdgeLabelSca}{\scalebox{\EdgeLabelScale}{{\rput{*0}(0,0){$ #1 $}}}}}}
\HideFrame
\HideGrid
\HideName
\RigidLabel
\FullState
\MediumState
\MediumPicture

\newcommand{\ThreeHalfChild}[2]{}%
\newcommand{\StateDiam}{1cm}
\newcommand{\path}{}
\newcommand{\ExtnF}[1]%
   {\overset{{\scriptscriptstyle \pmb{\smile}}}{#1}}

\newcommand{\DiffF}[1]%
   {\overset{{\scriptscriptstyle \pmb{\lor}}}{#1}}

\newcommand{\LocaF}[1]%
   {\overset{{\scriptscriptstyle \leftrightarrow}}{#1}}

\newcommand{\jsDist}[2][{}]%
   {\operatorname{\mathbf{d}_{#1}}\left(#2\right)}

\renewcommand{\lim}{{\operatornamewithlimits{\mathsf{lim}}}}

\newcommand{\Pfrak}{\mathfrak{P}}
\newcommand{\jsPart}[1]{{\operatorname{\Pfrak}\left(#1\right)}}

\newcommand{\x}{\! \times \!}

\newcommand{\SerSAnMon}[2]%
    {#1 \langle \! \langle  #2  \rangle \! \rangle }
\newcommand{\SerSAnMonD}[2]%
    {\left[#1\right] \langle \! \langle  #2  \rangle \! \rangle }
\newcommand{\SerMon}[1]%
    {\!\langle \! \langle  #1  \rangle \! \rangle }
\newcommand{\PolSAnMon}[2]%
    {{#1 \langle  #2 \rangle }}
\newcommand{\PolMon}[1]%
    {{\!\langle  #1 \rangle }}

\newsavebox{\LeftBraket}
\savebox{\LeftBraket}{\scalebox{0.7 1.2}{$<$}}
\newsavebox{\RightBraket}
\savebox{\RightBraket}{\scalebox{0.7 1.2}{$>$}}

\newcommand{\jsStar}[1]{{{#1}^{*}}}
\newcommand{\Ae}{\jsStar{A}}

\newcommand{\jsPlus}[1]{{{#1}^{+}}}
\newcommand{\Ap}{\jsPlus{A}}

\newcommand{\iotaK}{\iota_{\ShiftInd{K}}}

\newcommand{\compos}{\ccdot }

\newcommand{\phiikpsi}%
{{\varphi ^{-1}\! \compos        \iotaK \! \compos \! \psi }}
\newcommand{\phiiotpsi}[1]%
{{\varphi ^{-1}\! \compos        \iota _{\ShiftInd{#1}} \! \compos \! \psi }}
\newcommand{\phiintkpsi}[1]%
{{(#1\varphi ^{-1}\! \cap K) \psi }}

\newcommand{\jsless}
   {\mathrel{\leqslant_{\!\!\!\!\scriptscriptstyle{/}}}}
\newcommand{\jsgrea}
   {\mathrel{\geqslant_{\!\!\!\!\scriptscriptstyle{\backslash}}}}

\newcommand{\lexiconeq}
   {\preccurlyeq_{\!\!\!\!\!\scalebox{1.8 1}{\scriptscriptstyle{\pmb{/}}}}}

\newcommand{\jsAutUn}[1]%
   {\mbox{$\left\langle \thinspace #1 \thinspace \right\rangle $}}
\newcommand{\aut}[1]{\jsAutUn{#1}} 

\newcommand{\ShiftInd}[1]{\raisebox{-0.3ex}{$\scriptstyle{#1}$}}

\newcommand{\actb}{\mathbin{\raisebox{0.2ex}%
                        {${\scriptscriptstyle \circ} $}}}
\newcommand{\ccdot}{\actb} 

\newlength{\vbh}\newlength{\vbd}\newlength{\vbt}%
\newcommand{\CompAuto}[1]%
    {%
     \settodepth{\vbd}{\mbox{$\displaystyle{#1\strut}$}}%
     \settoheight{\vbh}{\mbox{$\displaystyle{#1\strut}$}}%
     \setlength{\vbt}{\vbh}\addtolength{\vbt}{\vbd}%
     {}%
     \psline[linewidth=0.8pt]{c-c}(0,-.65\vbd)(0,.9\vbh)%
     \hspace*{0.7pt}%
     {#1}%
     \kern0.8pt%
     \psline[linewidth=0.8pt]{c-c}(0,-.65\vbd)(0,.9\vbh)%
     }%

\newcommand{\bornedeuxlignes}[2]%
{\mbox{$
\begin{array}{c}{\scriptstyle #1}\\ {\scriptstyle #2} \end{array}
       $}}

\renewcommand{\path}[1]{\xrightarrow{\ #1 \ }} 

\newcommand{\ExpDer}[2][a]%
    {\operatorname{\frac{\partial}{\partial \mbox{$#1$}}}#2}
\newcommand{\ExpDerP}[2][a]%
    {\operatorname{\frac{\partial}{\partial\mbox{$#1$}}}\left(#2\right)}
\newcommand{\ExpDerr}[2][a]%
    {\operatorname{\frac{\partial_{\mathrm{R}}}{\partial \mbox{$#1$}}}#2}
\newcommand{\ExpDerB}[2][a]%
   {\operatorname{\frac{\partial_\mathsf{b}}{\partial \mbox{$#1$}}}#2}
\newcommand{\ExpDerBP}[2][a]%
   {\operatorname{\frac{\partial_\mathsf{b}}{\partial \mbox{$#1$}}}\left(#2\right)}

\renewcommand{\x}{\xmd\! \times \!\xmd}
\usepackage{stmaryrd}
\newcommand{\pqRep}[1]{\langle #1 \rangle_{\frac{p}{q}}}
\newcommand{\THFracs}[2]{\frac{#1}{#2}}%

\newcommand{\pqs}{\THFracs{p}{q}}

\newcommand{\Indpq}[1]{#1_{\pqs}}%
\newcommand{\Indtd}[1]{#1_{\tds}}%

\newcommand{\tds}{\THFracs{3}{2}}

\newcommand{\Tpq}{\Indpq{T}}

\newcommand{\Ttd}{\Indtd{T}}
\newcommand{\otd}{\Indtd{\omega}}
\newcommand{\obtd}{\Indtd{\widebar{\omega}}}

\newif\ifshowpicture
\newcommand{\ShowPicture}{\showpicturetrue}
\newcommand{\HidePicture}{\showpicturefalse}
\ShowPicture
\renewcommand{\val}[1]{\pi\hspace*{-0.1em}\left(#1\right)}
\newcommand{\realval}[1]{%
  \ifthenelse{\equal{#1}{}}%
  {\rho}
  {\rho\hspace*{-0.1em}\left(#1\right)}%
}
\newcommand{\realvalS}[1]{\rho(#1)}
\newcommandy[o]{\powerset}[1]{%
  \ifthenelse{\equal{#1}{}}%
  {\mathfrak{P}}
  {\mathfrak{P}\hspace*{-0.1em}\left(#1\right)}%
}
\newcommand{\minword}[1]{w^{-}_{#1}}
\newcommand{\maxword}[1]{w^{+}_{#1}}
\newcommand{\wlength}[1]{|#1|}
\newcommand{\bceil}[1]{\left\lceil#1\right\rceil}
\newcommand{\behav}[1]{L\hspace*{-0.1em}\left(#1\right)}
\newcommand{\maxletter}{\mathtt{maxLetter}}

\newcommand{\intint}[2]{\llbracket#1,#2\rrbracket}
\newcommand{\nspan}{\mathtt{span}}
\renewcommand{\mod}{\xmd\text{\scriptsize\bf \%}\xmd}
\newcommand{\modq}{\mod{q}}
\newcommand{\Span}{\mathsf{S}}
\newcommand{\Spanpq}{\Indpq{\Span}}
\newcommand{\Ape}{\Ap^{*}}
\newcommand{\Ao}{A^{\omega}}
\newcommand{\infbhv}[1]{\Lc\left(#1\right)}
\newcommand{\dgwm}{\mathrel{\ominus}}
\newcommand{\pqaux}{\frac{p}{q}}
\newcommand{\Spqaux}{\Indpq{\mathtt S}}
\newcommand{\Lpqaux}{\Indpq{L}}
\newcommand{\LpqBaux}{L_{\frac{p}{q},B}}
\newcommand{\Lrpqaux}{\Indpq{L}^{(r)}}
\newcommand{\Bpqaux}{B_{p,q}}
\newcommand{\Wpqaux}{\Indpq{W}}
\newcommand{\Wpqnaux}{W_{\frac{p}{q},n}}
\newcommand{\Wnaux}{W_{n}}
\newcommand{\Apaux}{A_p}
\newcommand{\Apsaux}{A_p^{\xmd *}}
\newcommand{\Apoaux}{A_p^{\xmd \omega}} 
\newcommand{\Aqoaux}{A_q^{\xmd \omega}} 
\newcommand{\Aqaux}{A_q}
\newcommand{\Aqsaux}{A_q^{\xmd *}} 
\newcommand{\Acaux}{\mathcal A}
\newcommand{\Tpqaux}{\Indpq{\mathcal{T}}}
\newcommand{\Dpqaux}{\Indpq{\mathcal{D}}}
\newcommand{\Mpqaux}{\Dpqaux}
\newcommand{\Thataux}{\widehat{\Tpqaux}}
\newcommand{\taupqaux}{\Indpq{\tau}}
\addStartTextModeZ{\pq}{\pqaux}
\addStartTextModeZ{\Mpq}{\Mpqaux}
\addStartTextModeZ{\Spq}{\Spqaux}
\addStartTextModeZ{\Lpq}{\Lpqaux}
\addStartTextModeZ{\LpqB}{\LpqBaux}
\addStartTextModeZ{\Bpq}{\Bpqaux}
\addStartTextModeZ{\Lrpq}{\Lrpqaux}
\addStartTextModeZ{\Wpq}{\Wpqaux}
\addStartTextModeZ{\Wpqn}{\Wpqnaux}
\addStartTextModeZ{\Wn}{\Wnaux}
\addStartTextModeZ[o]{\Ap}{\Apaux}
\addStartTextModeZ{\Aps}{\Apsaux}
\addStartTextModeZ{\Apo}{\Apoaux}
\addStartTextModeZ{\Aq}{\Aqaux}
\addStartTextModeZ{\Aqs}{\Aqsaux}
\addStartTextModeZ{\Aqo}{\Aqoaux}
\addStartTextModeZ[o]{\Ac}{\Acaux}
\addStartTextModeZ[o]{\Tpq}{\Tpqaux}
\addStartTextModeZ[o]{\That}{\Thataux}
\addStartTextModeZ[o]{\Dpq}{\Dpqaux}
\addStartTextModeZ{\taupq}{\taupqaux}

\newcommand{\rhythmaux}[1]{\widebar{#1}}
\addStartTextModeI{\rhythm}{\rhythmaux}

\newcommandx{\yaHelper}[2][1=\empty]{%
\ifthenelse{\equal{#1}{\empty}}%
  { \ensuremath{\scriptstyle{#2}}} 
  { \raisebox{ #1 }[0pt][0pt]{\ensuremath{\scriptstyle{#2}}}}  
}

\newcommandx{\yrightarrow}[4][1=\empty, 2=\empty, 4=\empty, usedefault=@]{%
  \ifthenelse{\equal{#2}{\empty}}
  { \xrightarrow{ \protect{ \yaHelper[#4]{#3} } } } 
  { \xrightarrow[ \protect{ \yaHelper[#2]{#1} } ]{ \protect{ \yaHelper[#4]{#3} } } } 
}

\newcommandy[o]{\pathx}[2][2={}]{
  \nlb%
  \ifthenelse{\equal{#2}{}}%
    {\yrightarrow{\ #1\ }[-1pt]}%
    {\underset{#2}{\xrightarrow{\ #1 \ }}}%
  \nlb%
}
\newcommand{\figur}[1]{Fig.\xmd\ref{f.#1}}

\newcommand{\theor}[1]{Theorem~\ref{t.#1}}
\newcommand{\secti}[1]{Sect.\xmd\ref{s.#1}}
\newcommand{\dex}[1]{(#1)}

\nolinenumbers

\ifdraft
  \linenumbers
  \usepackage[notref,notcite]{showkeys}
  
\fi

\usepackage{graphicx}
\usepackage{caption}
\usepackage{subcaption}
\usepackage{subfig}

\newlength{\vmabs}
\newlength{\vmord}

\SetVCDirectory{}
\ifdraft\ShowFrame

\fi
\HidePicture

\begin{document}
\tolerance=5000
\maketitle
\isodate

\begin{abstract}
This work is a contribution to the study of set of the representations of 
  integers in a rational base number system.
This prefix-closed subset of the free monoid is naturally represented 
  as a highly non regular tree whose nodes are the integers and whose 
  subtrees are all distinct. 
With every node of that tree is then associated a minimal infinite 
  word.

The main result is that a sequential transducer which computes for 
  all~$n$ the minimal word associated with~$n+1$ from the one 
  associated with~$n$, has essentially the same underlying graph as 
  the tree itself.  
  
These infinite words are then interpreted as representations of real 
  numbers; the difference between the numbers represented by these 
  two consecutive minimal words is the called the span of a
  node of the tree. 
The preceding construction allows to characterise the topological 
  closure of the set of spans.
\end{abstract}

\pagestyle{plain}     
\thispagestyle{plain} 

\section{Introduction}\lsection{intro}

The purpose of this work is a further exploration and a better
  understanding of the set of \emph{words} that represent integers in a
  rational base number systems.
These numeration systems have been introduced and studied
  in~\cite{AkiyEtAl08}, leading to some progress in the results around
  the so-called Malher's problem (\cf\cite{Mahl68}).
We give below a precise definition of rational base number
  systems and of the representation of numbers in such a system. 
But one can hint at the results established in this paper by just 
  looking at the figure showing the `representation tree' of the 
  integers -- that is, the compact way of describing the words that 
  represent the integers -- in a rational base number system 
  (\figur{rep-tre-sys}\dex{b} for the base~$\tds$) and by comparison 
  with the representation tree in a integer base number system   
  (\figur{rep-tre-sys}\dex{a} for the base~$3$).

Every subtree in the second tree is the full ternary tree whereas every 
  subtree in the first one is different from all other subtrees.
As a result, the language of the representations of the integers is 
  not only a non regular language, but the situation is even worse as 
  this language indeed satisfies no iteration 
  lemma of any kind (\cite{MarsSaka13b}).
With the hope of finding some order or regularity within what seems 
  to be closer to complete randomness (which, on the other hand, is not 
  established either) we consider the \emph{minimal words} originating 
  from every node of the tree.

In the case of an integer base, this is perfectly uninteresting: all 
  these minimal words are equal to~$0^{\omega}$.
In the case of a rational base these words are on the contrary all 
  distinct, none are even ultimately periodic (as the other infinite 
  words in the representation tree).
In order to find some invariant of all these distinct words, or at least a 
  relationship between them, we have studied the function that maps 
  the minimal word~$\minword{n}$ 
  associated with~$n$ onto the one associated with~$n+1$.
We tried to describe this function by a possibly infinite transducer.

  
 \ShowPicture
\begin{figure}[ht!]
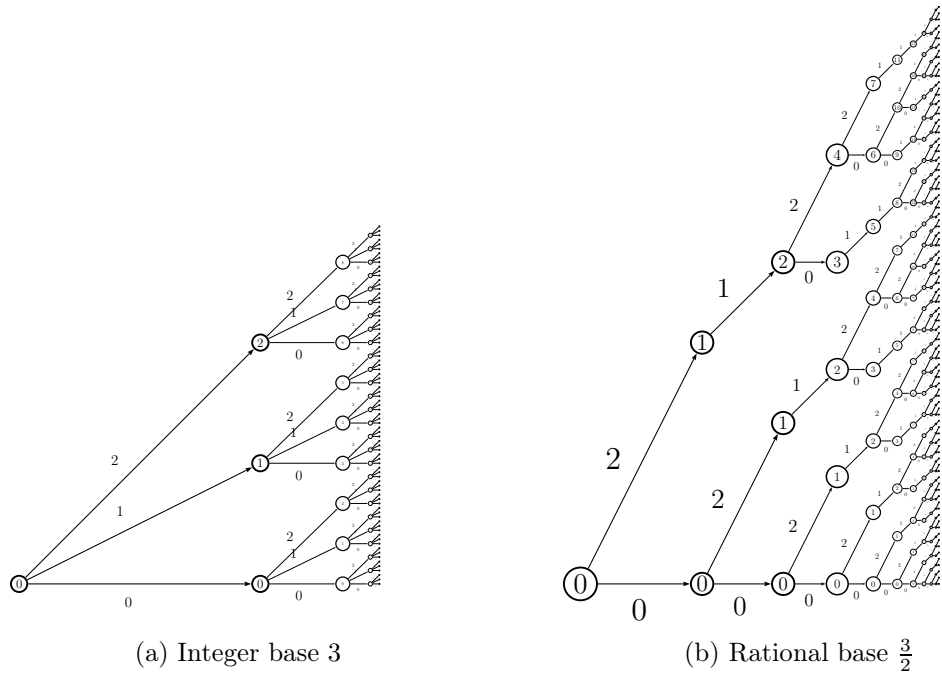

\FixVCScale{0.4}%
  \centering
  \begin{subfigure}{0.49\linewidth}
    \VCCall{p3_ro.tex}%
    \caption{Integer base~$3$}
  \end{subfigure}\ee%
  \begin{subfigure}{0.49\linewidth}
    \VCCall{p3q2_ro.tex}
    \caption{Rational base~$\tds$}
  \end{subfigure}
  \caption{Representation trees in two number systems}
  \lfigure{rep-tre-sys}
\end{figure}
 \ShowPicture

The computation of such a transducer in the case the base~$\tds$, and 
  more generally in the case of a base~$\pqs$ with~$p=2q-1$, leads to 
  a surprising and unexpected result.
The transducer, denoted by~$\Dpq$, is obtained by replacing in the 
  representation tree, denoted by~$\Tpq$, the label of every edge by a set 
  of pairs of letters that depends upon this label only.
In other words, the \emph{underlying 
  graphs} of~$\Tpq$ and~$\Dpq$ \emph{coincide}, and~$\Dpq$ is obtained 
  from~$\Tpq$ by a \emph{substitution} from the alphabet of digits into the 
  alphabet of pairs of digits, in this special and remarkable case.
  
The general case is hardly more difficult to describe, once it has 
  been understood. 
In the special case, the canonical digit alphabet has $p=2q-1$ 
  elements; in the general case, we still consider a digit alphabet 
  with~$2q-1$ elements denoted by~$\Bpq$, either by keeping the 
  larger~$2q-1$ elements of  
  the canonical digit alphabet, when~$p$ is is greater than~$2q-1$, 
  or by enlarging the canonical alphabet with enough negative digits, 
  when~$p$ is is smaller than~$2q-1$; in both cases, $p-1$ is the 
  largest digit. 
  
From~$\Tpq$ and with the digit alphabet~$\Bpq$, we then define 
  another `representation graph' denoted by~$\That$: either by 
  \emph{deleting the edges}  
  of~$\Tpq$ labelled by digits that do not belong to~$\Bpq$ in the 
  case where~$p>2q-1$ or, in the case where~$p<2q-1$ by \emph{adding 
  edges} labelled with the new negative digits. 
Then,~$\Dpq$ is obtained from~$\That$ exactly as above, by a 
  \emph{substitution} from the alphabet of digits into the  
  alphabet of pairs of digits.
This construction of~$\Dpq$, which we call the \emph{derived transducer}, 
  and the proof of its correctness are presented in \secti{der-tra}. 

In~\cite{AkiyEtAl08}, the tree~$\Tpq$, which is built from the representations 
  of integers, is used to \emph{define} the representations of real 
  numbers: the label of an infinite branch of the tree is the 
  development `after the decimal point' of a real number and the 
  drawing of the tree as a fractal object --- like in 
  \figur{rep-tre-sys} --- is fully justified by this point of view.
The same idea leads to the definition of the (renormalized%
\footnote{%
  The classical definition of span of the node~$n$ is, in the 
    fractal drawing, the width of the subtree rooted in~$n$. 
  This value is obviously decreasing (exponentially) with the depth of the
    node~$n$, hence the span of two nodes cannot be easily compared.
  In this article we only consider the \emph{renormalized span} 
    which is the span multiplied by~$(\frac{p}{q})^k$, 
    where~$k$ is the depth of the node~$n$.
})
  \emph{span} of a  node~$n$ of the representation tree:   
  it is the difference between 
  the real represented respectively by the maximal and the minimal 
  words originating in the node~$n$.
  
  
Again, this notion is perfectly uninteresting in the case of an 
  integer base~$p$: the span of node~$n$ is~always~$1$. 
And again, the notion is far more richer and complex in the case of a 
  rational base~$\pqs$.
The trivial relationship between the minimal word originating at 
  node~$n+1$ and the maximal word originating at node~$n$ leads to 
  the connexion between the construction of the derived 
  transducer~$\Dpq$ and the description of the set of 
  spans~$\Spanpq$. 
Not only the \emph{digit-wise difference} between maximal and minimal 
  words is written on the alphabet~$\Bpq$, but all these `difference 
  words' are infinite branches in the tree~$\That$. 
This is explained in \secti{spa-nod}.
From the structure of~$\That$, it then follows (\theor{spa-top}) 
  that the topological closure of~$\Spanpq$ is an interval in the 
  case where~$p<2q-1$, and a set with empty interior in the case 
  where~$p>2q-1$.
  
With every node~$n$ of the tree structure~$\Tpq$ is associated the 
  infinite minimal word~$\minword{n}$, an irregular infinite word 
  that looks as complex as the whole tree.
In conclusion, we have shown that a straightforward computation 
  of~$\minword{n+1}$ 
  from~$\minword{n}$ require the same structure 
  as~$\Tpq$ itself -- despite the fact that every minimal word 
  looks as complex as the whole tree -- whether it be performed directly 
  on the words, or 
  indirectly \via the span of the nodes.
It is this phenomenon that we call \emph{auto-similarity} of the 
  structure~$\Tpq$.
In this process, the value cases~${p=2q-1}$ appear to mark the boundary 
  between two different behaviour, in a more deeper way than that was 
  described in the first study of rational base number 
  systems~\cite{AkiyEtAl08}.
  
This paper is meant to be self-contained and gives, in particular, 
  all necessary definitions concerning rational base number systems. 
However, our paper~\cite{AkiyEtAl08} where these systems have been 
  defined and the sets of representations first studied will probably 
  be useful.  

\section{Preliminaries and notations}\lsection{prelim}

\subsection{Numbers and words}\lsection{words}

Given two \emph{real numbers}~$x$ and~$y$, 
  we denote by~$x/y$ or~$\frac{x}{y}$ their 
  division in~$\R$ (even if~$x$ or~$y$ happened to be integers),
  by~$[x,y]$ the corresponding interval of~$\R$ and by~$\ceil{x}$
  the integer~$n$ such that~$(n-1)<x\leq n$.
On the other hand, given two \emph{positive integers}~$n$ and~$m$, 
  we denote by~$n\div m$ and~$n \mod m$ respectively the quotient and 
  the remainder of the Euclidean division of~$n$ by~$m$, \ie ~$\msp {n=(n\div m)\xmd m + (n\mod{m})} \msp$ 
  and~$\msp 0\leq (n\mod{m}) < m$.
Additionally, we denote by~$\intint{n}{m}$ the integer 
  interval~$\set{n,(n+1),\ldots, m}$.


An \emph{alphabet} is a finite set of symbols called \emph{letters} 
  or \emph{digits} when they are integers.
Given an alphabet~$A$, we consider
both \emph{finite and infinite
  words} over~$A$ respectively denoted by~$\Ae$ and $\Ao$.
As in most cases letters will be digits, we denote the \emph{empty 
  word} by~$\varepsilon$.
For every positive integer~$p$, we denote by~$\Ap$ the canonical digit 
  alphabet of the base~$p$ number system:~${\Ap=\set{0, 1,\ldots,p-1}}$.
For clarity, we as much as possible denote finite words by~$u,v$ and 
  infinite words by~$w$.
The concatenation of two words~$u,v$ is either explicitly denoted by a low dot,
  as in $u.v$,
  or implicitly when there is no ambiguity, as in~$u\xmd v$.
A finite word~$u$ is said to be a prefix of a finite word~$v$ (resp. 
  an infinite word~$w$) if there exists a finite word~$v'$ (resp. an 
  infinite word~$w'$) 
  such that~$v=u\xmd v'$ (resp.~$w=u\xmd w'$).
The set of subsets of an alphabet~$A$ is denoted by~$\jsPart{A}$.

\subsection{Automata and transducers}\lsection{automata}

We deal here with a very special class of automata 
  and transducers only: they are infinite, their state set is~$\N$, 
  they are \emph{deterministic} (or \emph{letter-to-letter} and 
  \emph{sequential}), the initial state is~$0$, and all states are 
  final.

As usual, an \emph{automaton}~$\Xc$ over~$A$ is denoted by 
  a $5$-tuple~$\msp{\Xc=\aut{\N,A,\delta,0,\N}}$, where 
  $\msp{\delta:\N\x A\rightarrow \N}\msp$ is the \emph{transition 
  function}. 
The partial function~$\delta$ is extended to~$\N\x\Ae$, 
  and~$\delta(n,u)=m$ is also denoted by~$\msp{n\cdot u}=m\msp$ or 
  by~$n\pathx{u}m$.  
Given an integer~$n$, every state~$n\cdot a$ for some~$a$ 
  in~$A$ is called a \emph{successor} of~$n$.
A word~$u$ in~$\Ae$ (resp.~a word~$w$ in~$\Ao$) is \emph{accepted} 
  by~$\Xc$ if~$\msp{0\cdot u}\msp$ exists 
  (resp.~if~$\msp{0\cdot v}\msp$ exists for every finite prefix~$v$ of~$w$).
The \emph{language} of \emph{finite} words (resp. of \emph{infinite} 
  words) accepted by~$\Xc$ is denoted by~$\behav{\Xc}$ (resp. 
  by~$\infbhv{\Xc}$~). 


For transducers, we essentially use the notation of \cite{Bers79}, 
  adapted for the infinite case.
A \emph{transducer} is an automaton whose transitions are labelled by 
  pair of letters, it is formally a 
  tuple~${\Yc=\aut{\N,A \x B,\delta,\eta,0,\N}}$
  where~$\aut{\N,A,\delta,0,\N}$ is an 
  automaton, called \emph{the underlying input 
  automaton} of~$\Yc$,~$A$ is 
  called \emph{the input alphabet},~$B$ is the 
 \emph{output alphabet} and~$\eta\colon\N\x A\rightarrow B$ 
 is the \emph{output function}.
The transition function~$\delta$ is extended as in automata, and~$\eta$ is 
  as usual extended to~$\N\x A^* \rightarrow B^*$ 
  by~$\eta(n,\epsilon)=\epsilon$
  and~${\eta(n,u\xmd a)=\eta(n,u).\eta(n\cdot u,a)}$,
  and~$\eta(n,u)$ is also denoted by~$\msp n\ast u \msp$ for short.

Moreover, given two finite words~$u$ and~$v$, 
  we denote by~$n \pathx{\transpair{u}{v}} m$ the combination 
  of~$\msp{n\cdot u = m}\msp$ and~$\msp{n\ast u=v}\msp$.
We say that the \emph{image} 
  of a finite word~$u$ by~$\Yc$, 
  denoted by~$\Yc(u)$, is the word~$v$, if it exists, such 
  that~$0\pathx{\transpair{u}{v}}k$ for some~$k$. 
Similarly, the image of the infinite word~$w$ is~$w'$~if, for every 
  finite prefix~$u$ of~$w$,~$\Yc(u)$ is a prefix of~$w'$.

%
  
\subsection{Rational base number system}\lsection{rat-base}
Let~$p$ and~$q$ be two co-prime integers such that~$p>q>1$.
Given a positive integer~$N$, let us define~$N_0=N$ and, for 
all~$i>0$,
\begin{equation}
  q\xmd N_i = p\xmd N_{(i+1)} + a_i 
  \eqvrg
  \notag
\end{equation}
  where~$a_i$ is the remainder of the Euclidean 
  division of $q\xmd \N_i$ by~$p$,
  hence in~${\Ap=\intint{0}{p-1}}$.
Since~$p>q$, the sequence~$(N_i)_{i\in\N}$ is strictly decreasing and 
eventually stops at~$N_{k+1}=0$. 
Moreover, it holds that
\begin{equation}
  N = \sum^{k}_{i=0} \frac{a_i}{q} \left( \pq \right)^i 
  \eqpnt
  \notag
\end{equation}

The \emph{evaluation function}~$\pi$ is derived from this formula. 
Given a word~${a_na_{n-1}\cdots a_0}$ over~$\Ap$, and indeed over any 
  alphabet of digits,  its \emph{value} is defined by
\begin{equation}\label{eq.pi}
  \val{a_na_{n-1}\cdots a_0} = 
  \sum^{n}_{i=0} \frac{a_i}{q} \left( \pq \right)^i
  \eqpnt 
\end{equation}

Conversely, a word~$u$ in~$\Ape$ is called a~$\pq$-\emph{representation} of
  an integer~$x$ 
  if~${\val{u}=x}$.
Since the representation is unique up to leading 0's
  (see~\cite[Theorem~1]{AkiyEtAl08}),~$u$ is denoted by~$\cod{x}_{\pq}$
  (or~$\cod{x}$ for short) and 
  can be computed with the 
  modified Euclidean division algorithm above.
By convention, the representation of 0 is the empty word~$\epsilon$.
The set of $\pq$-representations of integers is denoted by~$\Lpq$:
\begin{equation}
	\Lpq = \Defi{\pqRep{n}}{n\in\N} \eqpnt
  \notag
\end{equation}

It should be noted that a rational base number system is \emph{not} 
  a~$\beta$-numeration --- where the representation of a number is 
  computed by the (greedy) R\'enyi algorithm 
  (cf.~\cite[Chapter~7]{Loth02}) --- in the special case
  where~$\beta$ is a rational number.
In such a system, the digit set is~$\set{0,1, \ldots, \ceil{\pq}}$
  and the weight of the~$i$-th leftmost
  digit is~$(\pq)^i$; whereas in the rational base number system, 
  they are~$\set{0,1~\ldots(p-1)}$ and~$\frac{1}{q}(\pq)^i$ respectively.


It is immediate that~$\Lpq$ is prefix-closed (since, in the modified Euclidean
  division algorithm~$\cod{N}=\cod{N_1}.a_0$) and prolongable
  (for every representation~$\cod{n}$, there exists (at least) an~$a$ 
  in~$\Ap$ such 
  that~$q$ divides~$(np+a)$ and 
  then~${\cod{\frac{np+a}{q}}=\cod{n}.a}$).
As a consequence,~$\Lpq$ can be represented as an infinite tree 
  (cf.~\rfigure{l32}).

\begin{figure}
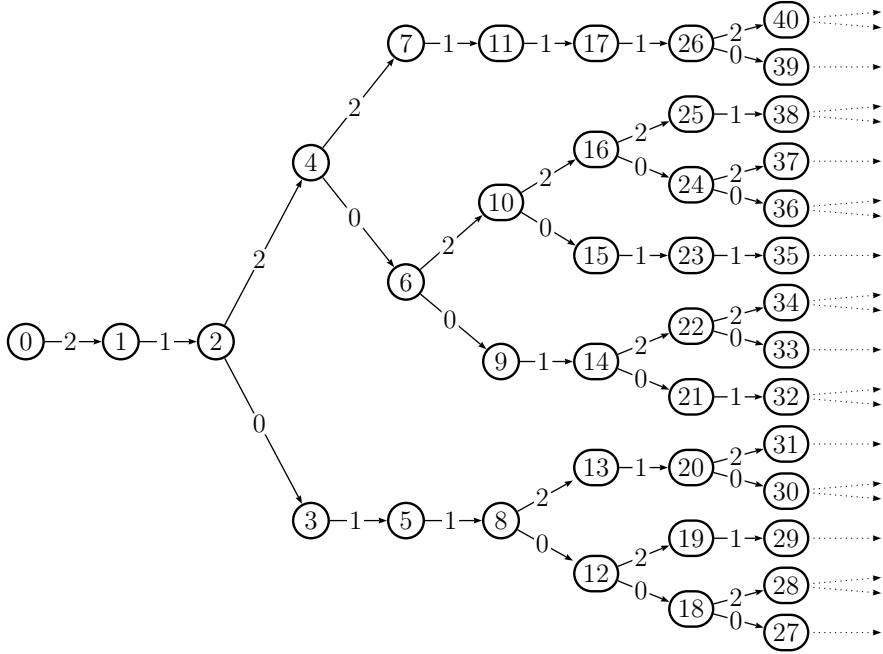

  \FixVCScale{0.5}%
  \setlength{\vmord}{1.25cm}
  \setlength{\vmabs}{2.5cm}
  \centering
  \VCCall{p3q2_L_leq40.tex}%
  \caption{The tree representation of the language~$L_{\frac{3}{2}}$}
  \lfigure{l32}
\end{figure}

It is known that \Lpq* is not rational (not even context-free), and
  the following automaton (in fact accepting the 
  language~$\msp0^*\Lpq$) is infinite.

\begin{definition}\ldefinition{tpq}%
Let~$\taupq\colon\N\x\Z\rightarrow\N$ be the (partial) function 
defined\/\footnote{
              The function $\taupq$ is  defined on~$\N\x\Z$ 
              instead of~$\N\x\Ap$ in anticipation of 
              future developments.
              }
 by:	
\begin{equation}\label{eq.tau}
  \forall n\in\N\quantvrg
    \forall a \in \Z\quantsp
	\taupq(n,a)=\left(\frac{n\xmd p+a}{q}\right) \ee 
    \text{if~$(n\xmd p+a)$ is divisible by~$q$.}
\end{equation}
We denote\/\footnote{%
     In \cite{AkiyEtAl08}, \Tpq* is denoted an infinite directed 
	 tree. 
     The labels of the (finite) paths starting from the root 
	 precisely formed the language~$\msp 0^*\Lpq$, as 
	 is~$\behav{\Tpq}$ in our case.} 
 by~\Tpq* the automaton~$\Tpq=\aut{\N,\Ap,\taupq,0,\N}$.
\end{definition}

In~$\Tpq$, we then have the transitions
  $\msp n\pathx{a}\left(\frac{n\xmd p+a}{q}\right)\msp$
  for every~$n$ in~$\N$, and every~$a$ in~$\Ap$ such that~$(np+a)$ is 
  divisible by~$q$.
The tree representation of \Lpq*, as in \rfigure{l32} augmented by
  an additional loop 
  labelled by~0 on the state~0 becomes a representation of~$\Ttd$.

\medskip

We call \emph{minimal alphabet} (resp. \emph{maximal alphabet}) the 
subalphabet~$\Aq=\intint{0}{(q-1)}$ (resp. the 
subalphabet~~$\intint{(p-q)}{(p-1)}$) of~$\Ap$.
Any letter of~$A_q$ is then called a \emph{minimal letter}, 
  \emph{maximal letter}
  being defined analogously.
The definition of~$\taupq$ implies that every state of \Tpq* has 
a successor by a \emph{unique} minimal (resp. maximal) letter.

\begin{definition}[minimal word]\ldefinition{minword}
  A minimal word (in the \pq*-system) is an infinite word in \Aqo*
    labelling an (infinite) path of \Tpq* (not 
    necessarily starting from the initial state~$0$).
\end{definition}
  
It is immediate that there exists a unique 
  infinite word in \Aqo* starting from the state~$n$ of \Tpq*. 
We call this word \emph{the} minimal word
    associated with~$n$ and denote it by~$\minword{n}$.
Additionally, we will use the term \emph{minimal outgoing label} of~$n$, 
  to designate
  the first letter of~$\minword{n}$ and \emph{minimal successor} 
  of~$n$ the unique
  successor of~$n$ by a minimal letter.
  
We define in a similar way the 
  \emph{maximal word}~$\maxword{n}$ associated with~$n$.

\section{The derived transducer}\lsection{der-tra}

The purpose of this section is to build an automaton 
over~${\Aq\x\Aq}$, that is, a \emph{letter-to-letter transducer}
  realising the function~$\minword{n}\mapsto\minword{(n+1)}$.
We call this transducer the \emph{derived transducer} and denote it by \Dpq*. 
It will be obtained from \Tpq* by a 
  local\footnote{The term \emph{local} is arguable see \rremark{locality}, in 
  the appendix.}
  transformation and this is the subject of \rsection{tpq->dpq}.

\subsection{From \Tpq* to \Dpq*}\lsection{tpq->dpq}
The transformation of \Tpq* into \Dpq* is a two-step process.
First, the structure of \Tpq* is changed locally, by changing the alphabet,
  and a new automaton \That* is thus obtained.
The second step consists in replacing the labels in \That* by a 
subset of~$\Aq\x\Aq$ by means of a \emph{substitution} (meaning that 
two transitions of \That* labelled by the same letter will be 
replaced by the same set of transitions) and produces \Dpq*.


\subsubsection{Changing the alphabet}
  We denote by \Bpq*  the alphabet~${\intint{p-(2q-1)}{(p-1)}}$.
%
In particular, if $p=(2q-1)$, $\Bpq = \Ap$; 
    if~${p<(2q-1)}$, \Bpq* contains negative digits;
    and if~${p>(2q-1)}$, \Bpq* is an uppermost subset of~$A_p$.
Note that \Bpq* is always of cardinal~$(2q-1)$, an \emph{odd number}, 
  that the digit~$(p-q)$ is then the \emph{centre} of~$\Bpq$ and 
  that its maximal element~$p-1$ coincides with the one of~$\Ap$. 
    
\medskip
The automaton~$\That$ is then defined by:
\begin{equation}
	\That=\aut{\N,\Bpq,\taupq,0,\N}
	\eqpnt
	\notag
\end{equation}
This is possible, even if~$\Bpq$ is larger than \Ap* because, 
  in \requation{tau},~$\taupq$ is defined 
  on~$\N\x \Z$, hence on~$\N\x\Bpq$.
  
\rfigure{l43->m43}, in the appendix, shows an example of the case when~$p$ 
  is (strictly) smaller than~$(2q-1)$, 
  i.e. one has to add edges (thicker arrows).
In this case, the resulting automaton is a DAG (more complex than a tree
  with one loop).
%
%
\rfigure{l73->m73} shows an example of the case when~$p$
  is (strictly) greater than~$(2q\nlb-\nlb1$), 
  i.e. one has to remove edges (dotted arrows).
In this case, the resulting automaton is a 
  forest (\ie an infinite union of trees).
%

\begin{figure}[ht!]
  \SmallPictures
  \begin{subfigure}{0.45\linewidth}
    \setlength{\vmord}{1.40cm}
    \setlength{\vmabs}{4cm}
    \centering
    \VCCall{p7q3_L-hat_leq16.tex} 
    \captionof{figure}{Transforming~$\Tc_\frac{7}{3}$ 
            into~$\widehat{\Tc_{\frac{7}{3}}}$} 
    \lfigure{l73->m73}
  \end{subfigure}\hfill
  \begin{subfigure}{0.45\linewidth}
    \setlength{\vmord}{1.40cm}
    \setlength{\vmabs}{4cm}
    \centering
    \VCCall{p7q3_D_leq16.tex}
    \caption{The derived transducer~$\Dc_{\frac{7}{3}}$}
    \lfigure{m73}
  \end{subfigure}
  \caption{From~$\Tc_\frac{7}{3}$ to $\Dc_{\frac{7}{3}}$}
\end{figure}

As already noted, if~$\msp{p=(2q-1)}$, $\Bpq = \Ap$ and \Tpq*=\That*.

\medskip

It is easy to verify that the process ensures that every state 
of~$\That$ congruent to~$-1$ modulo~$q$ has a unique successor and 
that all other states have exactly two successors. 
%

\subsubsection{Changing the labels}

Every label of \That* (which is a letter of~$\Bpq$) 
  is replaced by a \emph{set} of pairs of digits in~$\Aq\x\Aq$.
The \emph{label replacement 
  function}~${\omega_{\frac{p}{q}}\colon\Bpq\rightarrow \powerset{\Aq\x\Aq}}$
  (or~$\omega$ for short), 
  is more easily defined in two steps, as follows.
First, the function~$\widebar{\omega}$ computes the distance of the input to the 
  centre of~$\Bpq$: 
  $\msp\widebar{\omega}(a)=a-(p-q) \msp$, for every~$a$ 
  in~$\Bpq$.
Then, the image of~$a$ by~$\omega$ is the set of pairs of letters 
  in~$\Aq$ whose difference is~$\widebar{\omega}(a)$:
\begin{equation}
  \forall a \in \Bpq \quantsp 
    \omega(a) \e =\e \Defi{(\transpair{b}{c})\in\Aq\x\Aq}{c-b=\widebar{\omega}(a)}
	\eqpnt 
	\label{eq.omega}
\end{equation}

\begin{example}[The case~$\tds$]
	The functions~$\obtd$ and~$\otd$ are as follows:
	\begin{equation}
\begin{array}{rrcrc|crlcl}
\obtd\colon & 0 & \longmapsto & -1 & \e &
         \e & \otd\colon & 0 & \longmapsto & \set{~\transpair{1}{0}~} \\
\obtd\colon & 1 & \longmapsto &  0 & \e & 
         \e & \otd\colon & 1 & \longmapsto & \set{~\transpair{0}{0},~\transpair{1}{1}~} \\
\obtd\colon & 2 & \longmapsto &  1 & \e & 
         \e & \otd\colon & 2 & \longmapsto & \set{~\transpair{0}{1}~} 
   \end{array}
		\notag
	\end{equation}
and \figur{m32} shows~$\Dc_{\tds}$.
\end{example}

\begin{figure}[ht!]
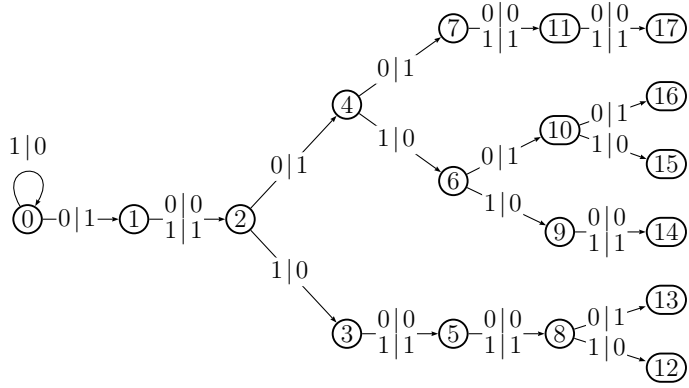

  \setlength{\vmord}{1cm}
  \setlength{\vmabs}{3.5cm}
\SmallPictures
\centering
  \VCCall{p3q2_D_leq40.tex} 
  \caption{The derived transducer~$\Dc_{\frac{3}{2}}$}
  \lfigure{m32}
\end{figure}
  
\rfigure{m73} shows the transducer~$\Dc_{\frac{7}{3}}$ and~$\Dc_{\frac{4}{3}}$ 
is represented by \rfigure{m43-app} in the appendix.





Formally, the transducer
  $\msp\Dpq=\aut{\N,\Aq\x\Aq,\delta,\eta,0,\N}\msp$
  is defined \emph{implicitly} or, more precisely, the transition 
  function~$\delta$ and the output function~$\eta$ are implicit 
  functions defined by the following statement:
\begin{multline}\label{eq.dpq}
  \forall n\in\N\quantvrg
    \forall a \in \Bpq\quantvrg
	  \forall (b,c) \in \omega(a) \quantsp\\
      \text{ $\taupq(n,a)$ is defined} \ee \Longrightarrow \ee n\pathx{\transpair{b}{c}}\taupq(n,a) \e
      \text{is a transition of~$\Dpq$,}\\
	  \text{that is, $\delta(n,b)=\taupq(n,a)$ and $\eta(n,b)=c$.}
\end{multline}

\medskip

In other words, the transitions of \Dpq* are labelled as follows:
if~${n\equiv-1~[q]}$, the state~$n$ has exactly one outgoing transition with 
    labels~${\transpair{0}{0},\,\transpair{1}{1},\,\ldots,\,\transpair{q-1}{q-1}}$.
Otherwise, the state~$n$ has two outgoings transitions and their labels 
    are~$\transpair{0}{k},\,{\transpair{1}{k+1}},\,\ldots,\,{\transpair{(q-k-1)}{q-1}}$ 
    for the upper transition 
    and~${\transpair{q-k}{0},\,{\transpair{{(q-k+1)}}{1}},\ldots,\,{\transpair{q-1}{k-1}}}$ 
    for the lower transition, 
    with~${k=a-(p-q)}$ and~$a$ being the maximal outgoing label of~$n$ in \Tpq*.

\medskip
    
The transducer constructed in this manner is 
  sequential and input-complete, as stated by the following lemma whose proof
  is given in the appendix.
  
\newcounter{lemma-seqic}
\setcounter{lemma-seqic}{\value{thm}}
\begin{lemma}
  For every state~$n$ of \Dpq* and every letter~$b$ of \Aq*, there exists
    a unique state~$m$ and a unique letter~$c$ 
    such that~${\msp n\pathx{\transpair{b}{c}}m}$.
\end{lemma}

\begin{corollary}
  For every infinite word~$w$ in \Aqo*,~$\Dpq(w)$ exists and is unique.
\end{corollary}

\subsection{Correctness of~$\Dpq$}\lsection{proof}

It remains to establish that \Dpq* has the expected behaviour,
  as stated in the following.
  
\newcounter{thmmpscor}%
\setcounter{thmmpscor}{\value{thm}}%
\begin{theorem}
	\ltheorem{mpq-correct}%
  For every~$n$ in~$\N$, $\msp\Mpq(\minword{n})=\minword{(n+1)}\msp$.
\end{theorem}
  
The proof of this theorem relies on the equivalent (and more 
  explicit) definition of the transition of \Dpq*, stated in the following
  proposition whose proof is given in the appendix.
  
\newcounter{propo-dpq-caract}%
\setcounter{propo-dpq-caract}{\value{thm}}%
\begin{proposition}
  \lproposition{dpq-caract}%
  If $\msp n \pathx{\transpair{b}{c}} m \msp$
  is a transition of~$\Dpq$,
  then
  \begin{equation}
    c = (b-(n+1)\xmd p)\mod{q}
  \ee\text{and}\ee
  m=\bceil{\frac{(n+1) \xmd p-b}{q}-1}
  \eqpnt
  \notag
  \end{equation}
\end{proposition}

In the case of finite words, 
  a stronger version can be stated.  
  
\begin{theorem}\ltheorem{mpq-cc}
  Given a base \pq* and two finite words~$u,v$ over \Aq*, 
    $\transpair{u}{v}$ labels a 
    run of \Mpq* 
    if, and only if there exists an integer~$n$ such that~$u$ is a prefix 
    of~$\minword{n}$ and~$v$ is a prefix of~$\minword{n+1}$.
\end{theorem}

This theorem is purposely stated on finite words, as a similar statement
  for infinite words would be false: 
  for every infinite word~$w$ of \Aqo*,~$\Dpq(w)$ exists, hence
  there is uncountably many pairs of infinite words~${\transpair{w}{\Dpq(w)}}$ accepted by \Mpq*
  while there is only countably many
  pairs~$\transpair{\minword{n}}{\minword{n+1}}$.
  
\section{Span of a node}\lsection{spa-nod}

In this part, 
  we consider the real value of infinite words.
We denote by~${\realval{}:A_p^{\xmd\omega}\rightarrow\R}$, the 
  \emph{real evaluation function}, 
  defined as follows:
\begin{equation}
  \realval{a_1a_2\cdots a_n\cdots} \e = \e 
    \sum_{i\geq0} a_i \xmd\left(\pq\right)^{-i} \eqpnt
\end{equation}

We denote by \Wpq* the language of infinite words~$\Lc(\Tpq)$.
It is proven in~\mbox{\cite[\nlb Theorem~2]{AkiyEtAl08}} that~$\realvalS{\Wpq}$ is the 
  interval~$[0,~\realval{\maxword{0}}]$.
By extension, we denote by \Wpqn* (or, for short, \Wn*) the language
  of infinite words~$\cod{n}^{-1}\Lc(\Tpq)$. 
Intuitively, an infinite word~$w$ over \Ap* is in \Wn* 
  if $\msp n\cdot u \msp$ exists in \Tpq* for every finite prefix~$u$ of~$w$.
Analogously to \Wpq*, the following holds.
\begin{lemma}
  For every integer~$n$, $\realvalS{\Wpqn}$ is the 
    interval~$[\realval{\minword{n}},~\realval{\maxword{n}}]$.
\end{lemma}

\begin{definition}
  For every integer~$n$, the span of~$n$, denoted by~$\nspan(n)$,
    is the size of~$\realvalS{\Wn}$:~$\nspan(n)=(\realval{\maxword{n}}-\realval{\minword{n}})$.
\end{definition}

Let~$a$ be a letter from the minimal 
  alphabet~${A_q=\intint{0}{(q-1)}}$ 
  and~$b$ a letter from the maximal alphabet~${\intint{(p-q)}{(p-1)}}$.
The integer~$(b-a)$ is necessarily in~$\intint{p-(2q-1)}{~p-1}=\Bpq$.
Hence, through this digit-wise subtraction, denoted 
  as~`$\dgwm$',~$(\maxword{n}\dgwm\minword{n})$ 
  is a word over~\Bpq*, and is called \emph{the span-word} of~$n$.
It is routine to check that the following statement is true.

\begin{lemma}
  For all integer~$n$, 
    $\nspan(n)=\realval{\maxword{n}\dgwm\minword{n}}$.
\end{lemma}

We denote by \Spq* the set of real numbers~$\set{\nspan(n)~|~n\in\N}$.
In order to establish properties of \Spq* (\rtheorem{spa-top}, below) we first
  need to consider span-words.
  
\begin{theorem}\ltheorem{that-complete}
  All span-words are accepted by \That*.
\end{theorem}

The proof of this theorem is a direct consequence 
  of~\rproposition{dpq->spq} below and requires more definitions.
The span-words is closely related to the derived transducer. 
There exists a (trivial) map~$m$
  from the minimal alphabet to the maximal alphabet, such that, 
  for all integer~$n$,~${m(\minword{n+1})=\maxword{n}}$.
\begin{align}
  m:\e \Aq\e  \longrightarrow \e & \intint{(p-q)}{(p-1)} \notag\\
     a \e  \longmapsto    \e & \maxletter(a+p)
\end{align}
where~$\maxletter(x)$ is the greatest integer congruent to~$x$ modulo~$q$ and 
  strictly smaller than~$p$.
By extending~$m$ to \Aqo*, \rtheorem{that-complete} is reduced to say that
  \That* accepts~$(m(\minword{n+1})\dgwm\minword{n})$ for every~$n$:

\newcounter{propo-dpq->spq}
\setcounter{propo-dpq->spq}{\value{thm}}
\begin{proposition}\lproposition{dpq->spq}
  If~$\transpair{w}{w'}$ is a pair of infinite words accepted by \Dpq*
  then \That* accepts the word~$(m(w')\dgwm w)$.
\end{proposition}

Analogously to the case of \Dpq*, \That* accepts uncountably many infinite
  words, therefore words that are not~$(\maxword{n}\dgwm\minword{n})$ for any~$n$.
That being said, it seems to be the best result we can hope, 
  as the following two corollaries hold.
\begin{corollary}
  Every finite word accepted by \That* is the prefix 
    of a span-word.
\end{corollary}
\begin{corollary}\lcorollary{topo-closure}
  The language of infinite words of \That* is the topological closure
    of the span-words.
\end{corollary}

In \cite{AkiyEtAl08}, it was hinted that there might be structural differences
  between two classes of rational base number systems.
Indeed, those where~${p \geq 2q-1}$  had an additional property, namely that
  every~$W_n$ contains at least two words (hence infinitely many).
It was however never proved that this property was false when~$p< 2q-1$.
The next statement provides a first element to differentiate these two classes
  of rational base number systems.
\newcounter{thmspatop}%
\setcounter{thmspatop}{\value{thm}}%

\begin{theorem}\ltheorem{spa-top}~
  \vspace*{-0.25em}
  \begin{enumerate}[label=\normalfont(\roman{*})]
    \item
      If~$(p < 2q)$, \Spq* is dense in~$[0,~\realval{\maxword{0}}]$.
    \item
      If~$(p> 2q)$, \Spq* is nowhere dense.
  \end{enumerate}
\end{theorem}

The proof of this theorem is not difficult but heavily relies on the notions and 
  properties developed in \cite{AkiyEtAl08}. 
We give a sketch of it in the appendix.

\section{Conclusion}

In the search of elucidating the structure of the set of representations of 
  integers in a rational base number system, we have shown that the 
  correspondence between
  two consecutive minimal words is achieved by a transducer that exhibits
  essentially the same structure as the one of the set of representations
  we started with.
We have called this property an ``\emph{auto-similarity}'' of the structure,
  as we have not shown that the structure is indeed \emph{self-similar}.
  
Let us note that the infinite transducer we have thus built realises the 
  correspondence for 
  \emph{all minimal words}. It is not a very good omen, but does not 
  contradict the following conjecture.

  \begin{conjecture}
  For every integer~$n$, there exists a finite transducer that 
    transforms~$\minword{n}$ into~$\minword{n+1}$.
\end{conjecture}

It is also remarkable that in this construction, the case~$p=2q-1$ appears as
  the frontier between two completely different behaviours of the system, in a
  much stronger way than it was described in our first work on rational base
  number systems.

%

\newcommand{\OneLetter}[1]{#1}


\clearpage
\newcounter{theortemp}
\setcounter{section}{0}
\renewcommand{\thesection}{A.\arabic{section}}
\renewcommand{\thefigure}{A.\arabic{figure}}

\begin{center}
    {\Large \textbf{Appendix}}
\end{center}
\vspace*{2em}
The section headings and numbers of the paper body are recalled, and 
prefixed with an~A, for an easier navigation.

\setcounter{section}{2}
\section{The derived transducer}
   \lsection{der-tra-app}%

   \subsection{From \Tpq* to \Dpq*}


   \subsubsection{Changing the alphabets}~
   
\begin{figure}[ht!] 
\setlength{\vmord}{1.25cm}%
\setlength{\vmabs}{2.5cm}%
  \FixVCScale{0.5}%
  \centering
  \VCCall{p4q3_L-hat_leq35.tex}
  \caption{Transforming~$\Tc_\frac{4}{3}$ 
           into~$\widehat{\Tc_{\frac{4}{3}}}$} 
  \lfigure{l43->m43}
\end{figure}  


   \subsubsection{Changing the labels}



\begin{figure}[ht!]
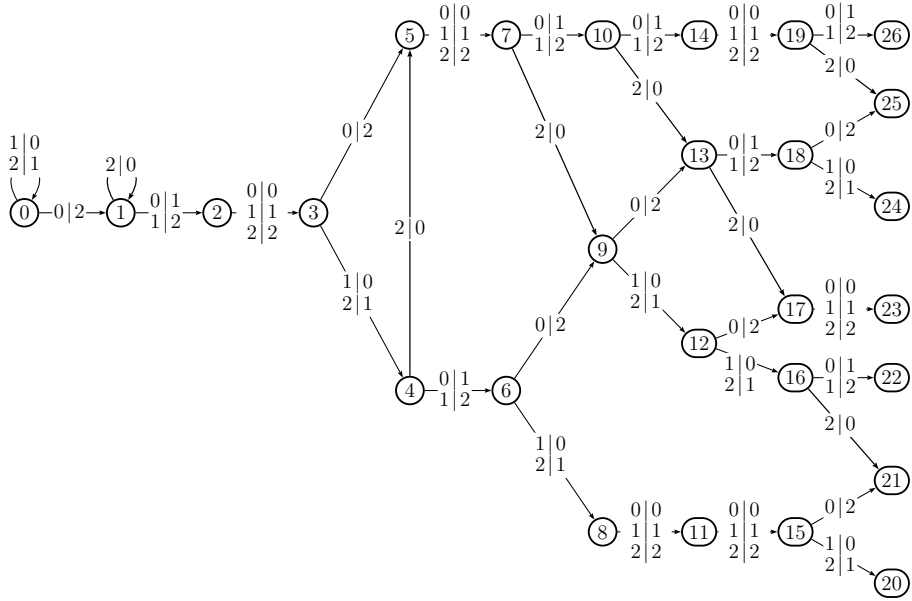

\setlength{\vmord}{2.33cm}
\setlength{\vmabs}{3.25cm}
  \centering
  \scalebox{0.65}{\VCCall{p4q3_D_leq35.tex}}
  \caption{The derived transducer~$\mathcal D_{\frac{4}{3}}$}
  \lfigure{m43-app}
\end{figure}

\setcounter{tmpthm}{\value{thm}}
\setcounter{thm}{\value{lemma-seqic}}

\begin{lemma}
  For all state~$n$ of \Dpq* and every letter~$b$ of \Aq*, there exists
    a unique state~$m$ and a unique letter~$c$ 
    such that~${\msp n\pathx{\transpair{b}{c}}m}$.
\end{lemma}
\setcounter{thm}{\value{tmpthm}}
\begin{proof} 
  \emph{Case where~$n$ is congruent to~$-1$ modulo~$q$:} by definition ~$n$
    has a unique successor associated with the letter~$(p-q)$, 
    $n \pathx{\omega(p-q)} \left(\frac{n\xmd p+a}{q}\right)$.
  In this case, the lemma's statement is immediate as~$\widebar{\omega}(p-q)=0$,
    hence~$\omega(p-q)$ is constituted of every pair~$\transpair{b}{b}$, 
    for every~$b$ in~$A_q$.
  
  \medskip
    
  \emph{Case where~$n$ is not congruent to~$-1$ modulo~$q$:} 
    let~$a$ be a maximal letter different than~$(p-q)$, 
    and~$d$ a minimal letter.
  It is sufficient to prove 
    that~${\msp\omega(a)\xmd\cup\xmd\omega(a-q)\msp}$ contains 
    exactly one pair 
    of the form~$\transpair{d}{e}$ for some~$e$.
    
  Since $\msp{\widebar{\omega}(a)=(a-p+q)}\msp$ 
    and~$\msp{\widebar{\omega}(a-q)=(a-p)}$,
    the difference between the two is~$q$, 
    hence at most one integer 
    of~${\set{~(d+\widebar{\omega}(a)),~(d+\widebar{\omega}(a-q))~}}$
    is in~$A_q$.
    
  Since~$a$ is a maximal letter, $\widebar{\omega}(a)$ is contained 
    in~$\intint{0}{(q-1)}$,
    as is~$d$, by definition of~$A_q$.
  Follows that~$\msp d+\widebar{\omega}(a) \msp$ is 
    in~$\intint{0}{2(q-1)}$, hence
    either~$\msp{d+\widebar{\omega}(a)}\msp$ is in $A_q$, 
    or it is in~$\intint{q}{2(q-1)}$,
    in which case~${(d+\widebar{\omega}(a-q))= (d+\widebar{\omega}(a)-q)}$ 
    is in~$A_q$.
\end{proof}

\subsection{Correctness of~$\Dpq$}
  \lsection{proof-app}

We establish now that \Dpq* has the expected behaviour, that is, we 
prove the main \theor{mpq-correct}
  as stated in the following.

\setcounter{tmpthm}{\value{thm}}
\setcounter{thm}{\value{thmmpscor}}

\begin{theorem}
	\ltheorem{mpq-correct.proof}%
  For every~$n$ in~$\N$, $\msp\Mpq(\minword{n})=\minword{(n+1)}\msp$.
\end{theorem}
  
\setcounter{thm}{\value{tmpthm}}
  
  
  
After the description of~$\Dpq$ by a transformation of~$\Tpq$, we 
characterise its transition and output functions by relations that 
will be used in further demonstrations.
\setcounter{tmpthm}{\value{thm}}
\setcounter{thm}{\value{propo-dpq-caract}}
\begin{proposition}
	\lproposition{dpq-caract.proof}%
  If $\msp n \pathx{\transpair{b}{c}} m \msp$
  is a transition of~$\Dpq$,
  then
  \begin{equation}
  	c = (b-(n+1)\xmd p)\mod{q}
	\ee\text{and}\ee
	m=\bceil{\frac{(n+1) \xmd p-b}{q}-1}
	\eqpnt
	\notag
  \end{equation}
\end{proposition}
\setcounter{thm}{\value{tmpthm}}

\begin{proof}
  If~$n \pathx{\transpair{b}{c}} m$, is a transition
    of \Dpq*, 
	then by hypothesis there exists a letter~$a$ in~$\Bpq$ such 
    that~$\transpair{b}{c}$ is
    in~$\omega(a)$, in which case:~$\widebar{\omega}(a)=c-b$, 
    hence~\fbox{$a=(p-q)+(c-b)$}.
    
  From~\requation{dpq}, we know that~$(n\xmd p +a)$ is 
    congruent to~$0$ modulo~$q$. 
  By replacing~$a$ with~$((p-q)+(c-b))$, we finally obtain that~$c$ is 
    congruent to~$(b-(n+1)\xmd p)$ modulo~$q$.
  Since~$c$ is in~$A_q$,~$c=(b-(n+1)\xmd p)\modq$
  
  From~\requation{dpq}, we know as well that~$m=\frac{n\xmd p + a}{q}$, 
    hence~${m=\frac{n\xmd p + ((p-q)+(c-b))}{q}}$, and after simplification,
    $\msp m=\frac{(n+1)\xmd p - b}{q}-\frac{c}{q}-1$.
  Since~$c$ is in~$A_p$, it is strictly smaller than~$q$, then
    $0\leq\frac{c}{q} < 1$ which concludes the proof.
\end{proof}

\rtheorem{mpq-correct} is then is a corollary of the next 
  proposition which describes the
  behaviour of \Mpq* starting from all states, not only the initial one.
\begin{proposition}\lproposition{mpq-correct}
  Let~$u$ and~$v$ be two words over \Aq*. 
  If~$n \pathx{u} m$ in \Tpq*
    and~$i \pathx{\transpair{u}{v}} j$ in \Dpq*, 
    then~${(n+i+1)\pathx{v}(m+j+1)}$ in \Tpq*.
\end{proposition}
\begin{proof}
  Let us first consider the special case where~$u$ is a single letter~$a$.
  The first hypothesis implies (from~\requation{tau}) 
    that~$m=\frac{n\xmd p+a}{q}$;
    the second (from~\rproposition{dpq-caract}) that~$v$ is 
    the single letter~$(a-(i+1)\xmd p)\mod{q}$
    and~${j=\ceil{\frac{(i+1) \xmd p-a}{q}-1}}$.
  
  It is routine to check that~$(a-(i+1)\xmd p)\mod{q}$ is 
    indeed an outgoing letter of~$(n+i+1)$.
  The successor in \Tpq* of~$(n+i+1)$ by this letter is
  
  \begin{align*}
    \frac{(n+i+1)p + (a-(i+1)\xmd p)\mod{q}}{q} \e = \e & \frac{ n\xmd p}{q}+\frac{(i+1)\xmd p}{q}+\frac{(a-(i+1)\xmd p)\mod{q}}{q} \\
                                              = \e & m+\frac{(i+1)\xmd p-a}{q} +\frac{(a-(i+1)\xmd p)\mod{q}}{q} \\
                                              = \e & m+\bceil{\frac{(i+1)\xmd p - a}{q}} \\
                                              = \e & m+j+1 \\
  \end{align*}
  The general case then consists in a simple induction over the length of~$u$.
\end{proof}
%
\begin{remark}[Locality]\lremark{locality}
  At the start of \rsection{der-tra}, it was claimed 
    that the transformation from \Tpq* to \Dpq* is local.
  Although it  undoubtedly is when~${p \geq (2q-1)}$,
    it is less clear when~$p < (2q-1)$.
    
  Indeed, at some point, one has to add an edge~$n\pathx{} (m-1)$ while having 
    access to the edge~$n\pathx{} m$, and must then access the state~$(m-1)$.
  Considering \Tpq* has an undirected graph, the path from~$(m-1)$ to~$n$ can be 
    arbitrarily large, which would contradict locality.
  However, we deemed it reasonable to have access to either a map from~$\N$ to the 
    states of \Tpq* 
    or simply a `decrementer' operator linking every state~$m$ to~$(m-1)$.
\end{remark}

\section{Span of a node}
\lsection{spa-nod-app}
  
\setcounter{tmpthm}{\value{thm}}
\setcounter{thm}{\value{propo-dpq->spq}}
\begin{proposition}\lproposition{dpq->spq.proof}
  If~$\transpair{w}{w'}$ is a pair of infinite words accepted by \Dpq*
  then \That* accepts the word~$(m(w')\dgwm w)$.
\end{proposition}
\begin{proof}
  It is enough to prove that, for every pair~$\transpair{b}{c}$ 
    in~$\omega(a)$,~$m(c)-b=a$.
  With this denotation, by definition 
    of \Dpq* (and more particularly~$\omega$, from 
    \requation{omega}),~${c=b+\widebar{\omega}(a)}$ and~${0\leq c < q}$,
    or more precisely
  $$\begin{array}{rcl}
    0     \leq& c                       &< q \\
    0     \leq& b + \widebar{\omega}(a) &< q \\
    0     \leq& b + a-(p-q)             &< q \\
    (p-q) \leq& (b+a)                   &< p \\
  \end{array}$$
  Therefore,~$(b+a)$ is a maximal letter, 
  hence~${\maxletter(b+a)=b+a}$, and finally~$m(c)-b=a$ 
  when replacing~$c$ and~$\widebar{\omega}$ by their expression.
\end{proof}
\setcounter{thm}{\value{tmpthm}}
\setcounter{tmpthm}{\value{thm}}
\setcounter{thm}{\value{thmspatop}}
\begin{samepage}
\begin{theorem}\ltheorem{spa-top.proof}~
  \vspace*{-0.25em}
  \begin{enumerate}[label=\normalfont(\roman{*})]
    \item
      If~$(p < 2q)$, \Spq* is dense in~$[0,~\realval{\maxword{0}}]$.
    \item
      If~$(p> 2q)$, \Spq* is nowhere dense.
  \end{enumerate}
\end{theorem}
\end{samepage}
\setcounter{thm}{\value{tmpthm}}

The proof of (i) essentially consists in the next Lemma and its corollary,
  stating that even though~$\That$ accepts words that \Tpq* doesn't, their
  values are redundant. 
  
\begin{lemma}\llemma{valT=valThat}
  If~$p<2q-1$, given a finite word~$u$ over~$B$ accepted by \That*, there exists
    a finite word~$v$ over \Ap* such that~$v$ is accepted by 
    \Tpq*,~$\val{u}=\val{v}$
    and~$\wlength{u}=\wlength{v}$.%
    \footnote{%
      The condition on $p$ and $q$ can be relaxed, but the case 
      where~$(p\geq2q-1)$ is trivial 
      and unnecessary in the following.
    }
\end{lemma}
\begin{proof} 
  Through a simple induction, one can reduce the statement to the special
    case where~$u$ is part of~$\Aps\xmd B$.
  We denote by~$n$ the non-negative integer~$\val{u}$. 
  It is then enough to prove that~$\wlength{\cod{n}}\leq\wlength{u}$, since 
    setting~${v=0^k\cod{n}}$ would satisfy both equations.
    
  We denote by~$u'$ (resp.~$v$) the word in \Aps* and by~$b$ (resp.~$a$) the 
    letter in~$B$ (resp. \Ap*) such that~$u=u'b$, (resp.~$\cod{n}=v'a$).
  
  Since~~~\begin{minipage}[t]{0.75\linewidth}
  \noindent 1. $v'$ is the representation of the integer~$\val{v'}$, \\
  \noindent 2. $\val{u'}$ is smaller than~$\val{v'}$, \\
  \noindent 3. $u'$ is in \Aps*; 
  \end{minipage}%
  $$\wlength{v'}\e \overset{1}{=}    \e \wlength{\cod{\val{v'}}} \e 
                   \overset{2}{\leq} \e \wlength{\cod{\val{u'}}} \e 
                   \overset{3} \leq \e \wlength{u'} \e$$
  hence~$\wlength{\cod{n}}\leq\wlength{u}$.
\end{proof}

\begin{corollary}\lcorollary{valT=valThat}
  If~$p<2q$, $\realval{\infbhv{\That}}=\realval{\infbhv{\Tpq}}$. %
  \footnote{%
    Here however the condition on~$p$ and $q$ is mandatory. 
  }
\end{corollary}
\begin{proof}[Proof of \rtheorem{spa-top}.(i)]
  Since~$\infbhv{\That}$ is the topological closure of the span words,
    and that (from \rcorollary{valT=valThat})~$\realval{\infbhv{\That}}=\realval{\infbhv{\Tpq}} = W_0$, the 
    set~$\set{\nspan(n)~|~n\in\N}$ is dense in~$W_0=[0,\realval{\maxword{0}}]$.
\end{proof}

The proof of \rtheorem{spa-top}.(ii) requires more notation. 
For all integer~$n$, we denote by~$W_n'$ the set of 
  words~$\msp0^*\cod{n}W_n$, \ie
  the (infinite) words of~$\Lc(\Tpq)$ whose run passes through the state~$n$.
In particular, with this notation, if~$n\pathx{u}m$, then~$W_m'\subseteq W_n'$, 
  which ``basically'' reduces \rtheorem{spa-top}.(ii) to the following 
  statement.

\begin{lemma}\llemma{cantor}
  If $p>2q$, for every integer~$n$, there exists an integer~$m$ 
    such that~$m$ is reachable from~$n$ in $\Tpq$ but not in \That*.
\end{lemma}
\begin{proof}
  We denote by~$S_i$ the set~$\set{n'~|~n\pathx{u}n'\text{ and } \wlength{u}=i}$.
  For all~$i$, $S_i$ is an integer interval, and 
    since~$p>2q$,~$\card{S_i}$ increases strictly with~$i$.
  It follows that~$S_{p+1}$ contains at least an integer~$m$ congruent to~$0$ 
    modulo~$p$ (beware, it is~$p$ and not~$q$).
  The state~$m$ is reachable in \Tpq* by a unique transition labelled by~$0$,
    and since~$0$ is not in~$B$ (because~$p>2q$),~$k$ 
    is not reachable in~$\That$.
\end{proof}

\begin{proof}[Proof of \rtheorem{spa-top}~(ii)]

  We denote by S the set $\set{\nspan(n)~|~n\in\N}$, and for all~$n$ 
    in~$\N$, we denote by~$W_n'$ the set of words~$\cod{n}W_n$. 
  
  Let us assume that~$S$ is dense in an 
    interval~$[x,y]$. 
  There exists a positive integer~$n$ such that~$\cod{n}\minword{n}$ 
    and~$\cod{n}\maxword{n}$ are both 
    in~$[x,y]$, hence~$S$ is dense in~$\realval{W_n'}$.
  From \rlemma{cantor}, there exists an integer~$m$ reachable from~$n$ in \Tpq*
    but not in \That*, hence no word of~$W_m'$ is accepted by \That*.
  The real values of theses words form a (non-trivial) sub-interval of~$\realval{W_n'}$.
  
  From~\cite[Corollary~38]{AkiyEtAl08} we know that every real number 
    has either one $\pq$-representation, or two, in which case one 
    is~$\cod{i+1}\minword{i+1}$ 
    and the other~$\cod{i}\maxword{i}$ for some~$i$.
  It implies that every word of the interior of~$\realval{W_{m}'}$
    has no $\pq$-representation outside of~$W_{m}'$, hence
    $\realval{W_n'}$ contains an open set whose intersection with~$S$ 
	is empty, a contradiction. 
\end{proof}
\end{document}